\documentclass[a4paper,11pt]{amsart} 
\usepackage{amsmath,amsxtra,amssymb,latexsym, amscd,amsthm}
\usepackage[dvipsnames]{xcolor}
\usepackage[mathscr]{eucal}
\usepackage{mathrsfs}
\usepackage{bbm}
\usepackage{enumerate}
\usepackage{pict2e}
\usepackage{tikz}
\usepackage{graphicx}
\usepackage[a4paper]{geometry}
\usepackage{color}

\usepackage{comment}
\usepackage{hyperref}
\numberwithin{equation}{section}

\theoremstyle{plain}
\newtheorem{thm}{Theorem}[section]

\newtheorem{prp}[thm]{Proposition}

\newtheorem{lem}[thm]{Lemma}
\newtheorem*{euc*}{Euclidean division}
\newtheorem*{fek*}{Fekete's Lemma}
\newtheorem*{kin*}{Kingman's Subadditive Ergodic Theorem}
\newtheorem*{fur*}{Furstenberg-Kesten Theorem}
\newtheorem*{ego*}{Egorov's Theorem}
\theoremstyle{definition}

\newtheorem{rem}[thm]{Remark}

\newtheorem*{rem*}{Remark}

\newcommand{\dd}{\mathrm{d}}
\newcommand{\ee}{\mathrm{e}}
\newcommand{\ii}{\mathrm{i}}
\renewcommand{\Im}{\operatorname{Im}}
\renewcommand{\Re}{\operatorname{Re}}

\newcommand{\N}{\mathbb{N}}
\newcommand{\Z}{\mathbb{Z}}

\newcommand{\R}{\mathbb{R}}
\newcommand{\C}{\mathbb{C}}

\newcommand{\T}{\mathbb{T}}

\newcommand{\LL}{\mathbb{L}}

\newcommand{\cN}{\mathcal{N}}
\newcommand{\cA}{\mathcal{A}}
\newcommand{\cC}{\mathcal{C}}

\newcommand{\fa}{\mathfrak{a}}
\newcommand{\fb}{\mathfrak{b}}

\DeclareMathOperator{\sgn}{sgn}

\DeclareMathOperator{\opn}{Op_N}

\DeclareSymbolFont{extraup}{U}{zavm}{m}{n}
\DeclareMathSymbol{\varheart}{\mathalpha}{extraup}{86}
\DeclareMathSymbol{\vardiamond}{\mathalpha}{extraup}{87}

\title{Quantum ergodicity for periodic graphs}
\author{Theo Mckenzie, Mostafa Sabri}
\address{Department of Mathematics, Harvard University, Cambridge, USA.}
\email{tmckenzie@math.harvard.edu}
\address{Department of Mathematics, Faculty of Science, Cairo University, Giza 12613, Egypt.}
\address{Science Division, New York University Abu Dhabi, Saadiyat Island, Abu Dhabi, UAE.}
\email{mostafa.sabri@nyu.edu}
\usepackage{calc}
\usepackage{graphicx}

\makeatletter
\newlength{\temp@wc@width}
\newlength{\temp@wc@height}
\newcommand{\widecheck}[1]{%
  \setlength{\temp@wc@width}{\widthof{$#1$}}%
  \setlength{\temp@wc@height}{\heightof{$#1$}}%
  #1\hspace{-\temp@wc@width}%
  \raisebox{\temp@wc@height+2pt}[\heightof{$\widehat{#1}$}]%
     {\rotatebox[origin=c]{180}{\vbox to 0pt{\hbox{$\widehat{\hphantom{#1}}$}}}}%
}
\makeatother

\begin{document}

\begin{abstract}
We prove quantum ergodicity for a family of periodic Schr\"odinger operators $H$ on periodic graphs. This means that most eigenfunctions of $H$ on large finite periodic graphs are equidistributed in some sense, hence delocalized. Our results cover the adjacency matrix on $\Z^d$, the triangular lattice, the honeycomb lattice, Cartesian products and periodic Schr\"odinger operators on $\Z^d$. The theorem applies more generally to any periodic Schr\"odinger operator satisfying an assumption on the Floquet eigenvalues.
\end{abstract}

\maketitle

\section{Introduction}

Our aim in this article is to prove quantum ergodicity for large periodic graphs. If $\Gamma_N$ is a sequence of finite graphs ``converging''\footnote{The pertinent convergence here is Benjamini-Schramm, but in the present context, we can just take $N$-balls around the origin and have $N\to\infty$.} to some infinite graph $\Gamma$, and if we study a Schr\"odinger operator $H_N = \cA_N + Q_N$ on $\ell^2(\Gamma_N)$, then quantum ergodicity is a \emph{delocalization} criterion stating that, in a weak sense, most eigenvectors of $H_N$ are equidistributed on the graph $\Gamma_N$. 

The terminology comes from \cite{CdV,Shni,Zel}, where the ergodicity of the geodesic flow on a compact manifold $M$ of unit volume (meaning the classical particle's free motion covers the manifold uniformly) is shown to imply a quantum counterpart of ergodicity, namely, the Laplacian wavefunctions $\psi_\lambda$ are equally likely to be anywhere on $M$ (more precisely $|\psi_{\lambda}(x)|^2\,\dd \mathrm{Vol}(x)$ approaches the uniform measure $\dd \mathrm{Vol}(x)$, when $\lambda$ gets large).

Quantum ergodicity for large regular graphs with few cycles was first proved in \cite{ALM}, for the adjacency matrix $H_N=\cA_N$. In this case the limiting graph $\Gamma$ was the $(q+1)$-regular tree $\T_q$. The general case where $\Gamma$ is an infinite tree which is not necessarily regular and $H_N=\cA_N+Q_N$ was later established in \cite{AS}, assuming $H_\Gamma$ has absolutely continuous spectrum. This includes regimes of the Anderson model \cite{AS2}, as well as ``periodic trees with periodic potentials'', more precisely universal covers of finite graphs \cite{AS3}.

The previous results also required the $\Gamma_N$ to be expanders. It was shown in \cite{Theo} that counterexamples exist if expansion is dropped. Examples of regular expander graphs violating quantum ergodicity were also constructed in \cite{Theo}; there $\Gamma$ was no longer a tree. However, it remained open whether more specific families of graphs satisfy quantum ergodicity.

In this paper we show that quantum ergodicity is in fact satisfied for a large family of non-tree graphs $\Gamma$, namely graphs which are periodic with respect to a basis of $\Z^d$. The simplest example is the adjacency matrix on $\Z^d$, but the results also apply to some classes of Schr\"odinger operators with periodic potentials on various lattices. These graphs do not satisfy the expansion or tree properties of previous proofs. Therefore we need new, different techniques to solve the problem in this case. 

On $\Z^d$, if we consider the eigenbasis $e_r^{(N)}(k)=\frac{1}{N^{d/2}}\ee^{2\pi\ii k\cdot r/N}$ for the adjacency matrix on a sub-cube $\Gamma_N$ of sidelength $N$ with periodic boundary conditions, we note that $|e_r^{(N)}(k)|^2 = \frac{1}{N^d}$ is perfectly uniformly distributed on $\Gamma_N$. Similarly, for a periodic Schr\"odinger operator $H$ on a periodic lattice, the Bloch theorem ensures that for any $\lambda\in \sigma(H_N)$, we can find an eigenfunction $\Psi_\lambda$ such that $|\Psi_{\lambda}(x)|^2$ is a periodic function (see \S\ref{sec:bloch} for a discussion and a proof of this result in our context). Here we study whether such delocalization is satisfied \emph{for any eigenbasis} of the Schr\"odinger operator.

By virtue of their homogeneity, it is quite intuitive to expect delocalization on periodic lattices. Indeed, the spectrum is generally absolutely continuous, though flat bands (infinitely degenerate eigenvalues) can appear \cite{KorSa}. The dynamics is also ballistic \cite{BoSa}, meaning the waves spread at maximum speed with time. Here we show that from a spatial point of view, the behavior is quite rich~:
\begin{itemize}
\item There is a simple family of periodic graphs which is quantum ergodic, i.e. the probability measure $\sum_{x\in \Gamma_N} |\psi_u^{(N)}(x)|^2\delta_x$ is close to the uniform measure $\frac{1}{|\Gamma_N|}\sum_{x\in \Gamma_N}\delta_x$, for most $u$. See Theorem~\ref{thm:mainnu1}. Here $(\psi_u^{(N)})_u$ is an orthonormal eigenbasis.
\item In another class of periodic Schr\"odinger operators, we have partial quantum ergodicity, in the sense that we no longer have $|\psi_u^{(N)}(x)|^2\approx \frac{1}{|\Gamma_N|}$, but the sum of $|\psi_u^{(N)}(x)|^2$ over any periodic block is approximately the same (Theorem~\ref{thm:maingen}, Proposition~\ref{prp:cart}). This means that $\psi_u^{(N)}$ does not favor any particular block, but the mass of $\psi_u^{(N)}$ may not be uniform within the block.
\item In other classes of periodic graphs, quantum ergodicity fails completely (\S\ref{sec:deco} and \S\ref{sec:tenstro}).
\end{itemize}

Examples of these three types are $\cA$ on $\Z^d$, on an infinite cylinder (Fig.~\ref{fig:cyl}), and on the graph in Figure~\ref{fig:box}, respectively.

\subsection{Main results}
Let $\Gamma$ be a connected, locally finite graph in some Euclidean space. We assume $\Gamma$ is invariant under translations of some linearly independent vectors $\fa_1,\dots,\fa_d$.

If $\cC_{\fa} = \{x_1\fa_1+\dots+x_d \fa_d:x\in [0,1)^d\}$ is the basic cell, then
\[
V_f = V\cap \cC_{\fa} = \{v_1,\dots,v_\nu\}
\]
is the unit crystal, containing $\nu$ vertices. Given $x = (x_1,\dots,x_d)\in \R^d$, denote $x_{\fa} = \sum_{i=1}^d x_i \fa_i$. The graph $\Gamma$ then consists of periodic $V_f$ blocks of size $\nu$, in the sense that
\begin{equation}\label{e:sumlat}
V = \Z_\fa^d + V_f\,,
\end{equation}
where $\Z_\fa^d = \{n_\fa : n\in \Z^d\}$. Any $u\in V$ thus takes the form $u = \lfloor u\rfloor_{\fa}+\{u\}_{\fa}$, where $\lfloor u\rfloor_{\fa}\in\Z_\fa^d$ and $\{u\}_{\fa}\in V_f$ represent the integer and fractional parts of $u$, respectively. 

In the case of the simple lattice $V = \Z^d$ we have $\fa_j = \mathfrak{e}_j$ the standard basis and $V_f = \{0\}$. In general one can view \eqref{e:sumlat} as expressing the vertex set $V$ as $\nu$ copies of the sub-lattices $\Z_{\fa}^d$, shifted by vertices $v_n\in V_f$.

We consider a Schr\"odinger operator $H = \cA + Q$ on $\Gamma$, where $\cA$ is the adjacency matrix and $Q$ satisfies
\[
Q(v+\fa_i) = Q(v)
\]
for $v\in\Gamma$ and $i=1,\dots,d$. The potential $Q$ is thus periodic with at most $\nu$ values.

Let $\Gamma_N = \Gamma \cap \{x_1\fa_1+\dots+x_d\fa_d:x\in [0,N)^d\}$. Let $H_N$ be the Schr\"odinger operator defined analogously on $\Gamma_N$. We endow $\Gamma_N$ with periodic boundary conditions~: if $\psi\in \ell^2(\Gamma_N)$, then $\psi(v+N\fa_j) := \psi(v)$. Our first result is the following.

\begin{thm}[Case $\nu=1$]\label{thm:mainnu1}
Let $\psi_u^{(N)}$ be an orthonormal basis of $\ell^2(\Gamma_N)$ consisting of eigenvectors of $H_N$. Suppose the fundamental crystal has only one vertex, $V_f=\{o\}$. Then for any observable $a=a_N :\Gamma_N \to \C$ such that $|a_N(v)|\le 1$ for all $v$ and $N$, we have
\begin{equation}\label{e:mainnu1}
\lim_{N\to\infty} \frac{1}{|\Gamma_N|}\sum_{u\in \Gamma_N} \left| \langle \psi_u^{(N)}, a\psi_u^{(N)}\rangle - \langle a\rangle\right|^2 = 0\,,
\end{equation}
where $\langle \psi_u^{(N)}, a\psi_u^{(N)}\rangle = \sum\limits_{v\in\Gamma_N} |\psi_u^{(N)}(v)|^2a(v)$ and $\langle a\rangle = \frac{1}{|\Gamma_N|}\sum\limits_{v\in\Gamma_N} a(v)$ is the uniform average.
\end{thm}

This means that in a weak sense, we have $|\psi_u^{(N)}(v)|^2\approx \frac{1}{|\Gamma_N|}$ when $N$ is large enough. That is, the eigenvectors $\psi_u^{(N)}$ are uniformly distributed. This theorem applies to the adjacency matrix on $\Z^d$ and the triangular lattice, see \S\ref{sec:scalexa} for more graphs.

To our knowledge, Theorem~\ref{thm:mainnu1} is the first positive result establishing quantum ergodicity for a general family of graphs $\Gamma$ having cycles.

This statement is generally false for higher $\nu$, quantum ergodicity can be completely violated for $\nu=2$ without further assumptions, see \S\ref{sec:stro}. For general $\nu$, we make an assumption on the Floquet eigenvalues and relax the conclusion. Let $\fb_1,\dots,\fb_d$ be the dual basis, satisfying $\fa_i\cdot\fb_j=2\pi\delta_{i,j}$. We similarly denote $\theta_\fb = \sum_{i=1}^d\theta_i\fb_i$. Then we have:

\begin{thm}[General case]\label{thm:maingen}
Let $\psi_u^{(N)}$ be an orthonormal basis of $\ell^2(\Gamma_N)$ consisting of eigenvectors of $H_N$. Let $H(\theta_\fb)$ be the $\nu\times \nu$ matrix arising in the Floquet decomposition, with eigenvalues $E_1(\theta_\fb),\dots,E_\nu(\theta_\fb)$. Suppose that for any $s,w\in \{1,\dots,\nu\}$, we have
\begin{equation}\label{e:eigenass}
\lim_{N\rightarrow \infty}\sup_{\substack{m\in \LL_N^d\\m\neq 0}} \frac{\#\{r\in \LL_N^d : E_s(\frac{r_\fb+m_\fb}{N})-E_w(\frac{r_\fb}{N})=0\}}{N^d} = 0\,,
\end{equation}
where $\LL_N = \{0,1,\dots,N-1\}$. Then,

\begin{enumerate}[\rm(i)]
    \item For any observable $a_N :\Gamma_N \to \C$ such that $|a_N(v)|\le 1$ for all $v$ and $N$, we have
\begin{equation}\label{e:main}
\lim_{N\to\infty} \frac{1}{|\Gamma_N|}\sum_{u\in \Gamma_N} \left| \langle \psi_u^{(N)}, a\psi_u^{(N)}\rangle - \langle \psi_u^{(N)}, \opn(\overline{a})\psi_u^{(N)}\rangle\right|^2 = 0\,,
\end{equation}
where $\opn(\overline{a})$ is an explicit operator (see \eqref{e:opnabarpsi}).

If $a=a_N$ is real-valued, we have
\begin{equation}\label{e:asavboun}
\min_{v_q\in V_f} \langle a(\cdot+v_q)\rangle \le \langle \psi_u^{(N)}, \opn(\overline{a})\psi_u^{(N)}\rangle \le \max_{v_q\in V_f} \langle a(\cdot+v_q)\rangle\,,
\end{equation}
where $\langle a(\cdot+v_q)\rangle = \frac{1}{N^d}\sum_{k\in\LL_N^d} a(k_\fa+v_q)$.

\item
If $a$ is locally constant, in the sense that it takes a constant value on each periodic block, then it suffices to assume \eqref{e:eigenass} for $w=s$, $1\le s\le \nu$. Moreover, in this case,
\begin{equation}\label{e:avconverge}
\langle \psi_u^{(N)}, \opn(\overline{a})\psi_u^{(N)}\rangle = \langle a\rangle:=\frac{1}{|\Gamma_N|}\sum_{v\in\Gamma_N} a(v).
\end{equation} 
Specifically, this is true if $\nu=1$. 

\end{enumerate}

\end{thm}

This theorem applies, for example, to the adjacency matrix on the honeycomb lattice and to periodic Schr\"odinger operators on $\Z$, see \S\ref{sec:honey} and \S\ref{sec:1d}.

Assumption \eqref{e:eigenass} says in particular that the Floquet eigenvalues should not have a short period and should not ``hesitate'' while tracing the band, going back and forth too often at exactly the same speed. More precisely, for any nonzero $\alpha$ and any $s$, the set
\begin{equation}\label{e:aalphas}
A_{\alpha,s}:=\{\theta\in [0,1)^d:E_s(\theta_\fb+\alpha_\fb) = E_s(\theta_\fb)\}
\end{equation}
should be of zero measure. For example, for $d=1$, we should not have $E_s(\theta_\fb) = \cos 4\pi\theta$, as then for $\alpha=\frac{1}{2}$, we get $E_s(\theta_\fb+\alpha_\fb)=E_s(\theta_\fb)$ for all $\theta$. The assumption also implies there is no point spectrum, because flat bands require $E_s$ to be constant for some $s$.

In an earlier version of the manuscript we left as an open problem whether assumption \eqref{e:eigenass} is satisfied for Schr\"odinger operators on $\Z^d$ with a periodic potential. This has since been solved by Wencai Liu \cite{Wen2} using results on the irreducibility of Bloch varieties. See \S\ref{sec:irred} for background and further criteria. In particular, \cite{Wen2} and Theorem~\ref{thm:maingen} imply that Schr\"odinger operators with periodic potentials on $\Z^d$ are quantum ergodic for any $d$.

Theorem~\ref{thm:mainnu1} only applies to the adjacency matrix on regular graphs of even degree (as follows from the assumption $\nu=1$, see \S\ref{sec:scalfib}). The following proposition uses Theorem~\ref{thm:maingen} to provide concrete applications to non-regular graphs endowed with a periodic potential.

\begin{prp}[Cartesian products]\label{prp:cart}
Suppose that $\Gamma$ is a $\Z_\fa^d$-periodic graph with $\nu=1$, and let $G_F$ be any finite graph, endowed with some potential $Q$. Then the Cartesian product $\Gamma\mathop\square G_F$ is a periodic graph with fundamental crystal $V_f=G_F$ and periodic potential $Q$ copied across the $G_F$ layers. Moreover, assumption \eqref{e:eigenass} is satisfied, so \eqref{e:main} holds true.

If for $\Gamma\mathop\square G_F$, the orthonormal basis is of the form $\psi_{n,j} = \phi_n\otimes w_j$, where $(\phi_n)$ is an orthonormal eigenbasis for $H_{\Gamma_N}$, and $(w_j)$ is an orthonormal eigenbasis for $H_{G_F}$, then
\begin{equation}\label{e:cartav}
\langle \psi_{n,j},\opn(\overline{a})\psi_{n,j}\rangle = \sum_{v_q\in G_F}\langle a(\cdot+v_q)\rangle |w_j(v_q)|^2\,,
\end{equation}
where $\langle a(\cdot+v_q)\rangle = \frac{1}{N^d}\sum_{k\in \LL_N^d} a(k_\fa+v_q)$.
\end{prp}

Theorem~\ref{thm:maingen}(ii) shows that for most $u$, $|\psi_u^{(N)}|^2$ behaves as a periodic function across the blocks, but the distribution of its mass within each block may be non-uniform. Very loosely speaking, one has the picture that most eigenfunctions behave like Bloch functions. More precisely, for most $u$, $\sum_{v_q\in V_f}|\psi_u^{(N)}(k_\fa+v_q)|^2\approx \frac{1}{N^d}$, for any $k$ (in a weak sense).

On the other hand, \eqref{e:cartav} shows that the mass distribution within each block is not universal and can depend on the eigenbasis in general (see \S\ref{sec:cyl} for a concrete example). Such base-dependence never appeared in the tree models of \cite{Ana,AS}. There, the theorems established that $|\psi_u^{(N)}(v)|^2 \approx \frac{\Im \widetilde{g}_N^{\lambda_j}(\tilde{v},\tilde{v})}{\sum_{v\in \Gamma_N} \Im \widetilde{g}_N^{\lambda_j}(\tilde{v},\tilde{v})}$, where $\widetilde{g}_N^z$ is the Green's function of the universal cover of $\Gamma_N$. In particular, it is certainly independent of $\psi_{\lambda_j}^{(N)}$. Here we have a different phenomenon which can be regarded as partial quantum ergodicity.

Such partial quantum ergodicity can be violated even in dimension one~:

\begin{prp}\label{prp:zno}
There exist $\Z$-periodic graphs which violate \eqref{e:main}.
\end{prp}

We give examples in \S\ref{sec:deco} and \S\ref{sec:tenstro}. These graphs have point spectrum and \eqref{e:eigenass} is not satisfied. It is natural to ask if assumption \eqref{e:eigenass} can be dropped if we simply assume that $H_\Gamma$ has pure ac spectrum. We construct a counterexample in \S\ref{sec:tenstro}~:

\begin{prp}\label{prp:counterass}
There exist periodic graphs with purely absolutely continuous spectrum which are not quantum ergodic.
\end{prp}

Our counterexample violates \eqref{e:eigenass} for some $w\neq s$. Assumption \eqref{e:eigenass} is satisfied for all $w=s$, $1\le s\le \nu$. In view of Theorem~\ref{thm:maingen}(ii), quantum ergodicity is thus satisfied for locally constant observables. The interpretation is that, while the total mass within each periodic block is the same, if we focus on just one vertex $v_1$ of the block and see how the mass $|\psi_u(k+v_1)|^2$ varies with the position $k$, then the variation is not uniform.

We do not have a counterexample violating \eqref{e:eigenass} for $w=s$. The following question on the eigenvalues seems to be of independent interest. Recall $A_{\alpha,s}$ in \eqref{e:aalphas}.

\subsection*{Open problem}
Do all connected periodic graphs with purely absolutely continuous spectrum satisfy $\mathrm{Leb}(A_{\alpha,s})=0$, for any $\alpha\neq 0$ and any $1\le s\le \nu$ ?

\medskip

If the answer is yes, then quantum ergodicity \emph{for locally constant observables} holds for all connected periodic graphs with pure ac spectrum, by Theorem~\ref{thm:maingen}(ii) and the argument in the proof of \cite[Cor. 1.3]{Wen2}. See Remark~\ref{e:perpluspot} for a related question.

\begin{rem}
Instead of considering the whole spectrum in Theorem~\ref{thm:maingen}, we can instead suppose that \eqref{e:eigenass} is satisfied in some interval $I$, then the conclusion \eqref{e:main} now holds if we average over $\lambda_u^{(N)}\in I$ instead of $u\in \Gamma_N$. This is similar to what is done in \cite{AS} for the high girth regime.  In other words, if part of the spectrum is well-behaved, then the corresponding eigenfunctions are quantum ergodic. This is helpful for example for graphs having flat bands but satisfying \eqref{e:eigenass} away from the degenerate eigenvalue. Then our theorem applies to these regions. For the technical details, see Remark~\ref{rem:I}.
\end{rem}

\begin{rem}[Convergence rate]
The proof shows that the variance on the LHS of \eqref{e:main} is essentially bounded from above by the fraction in \eqref{e:eigenass}. For $\nu=1$, we bound the latter by $\frac{C}{N}$ in \S\ref{sec:scalfib}, so the speed of convergence is at least $\frac{1}{N}$ in Theorem~\ref{thm:mainnu1}, which is significantly faster than the logarithmic rate $\frac{1}{\log N}$ of the tree case \cite{Ana,ALM,AS}.
\end{rem}

\begin{rem}\label{e:perpluspot}
The fact that a perfectly homogeneous graph like the one in Figure~\ref{fig:box} supports localized eigenfunctions is quite counterintuitive. It seems the problem is not the geometry here, but the potential $Q=0$, which is a degenerate case if one regards periodic graphs as naturally carrying a periodic potential. It is conjectured in \cite{KorSa} that the spectrum of any periodic graph becomes absolutely continuous if we add a potential $Q(v_1)<\dots<Q(v_\nu)$ having distinct values in the fundamental cell. In this spirit, we show that the graph in Figure~\ref{fig:box} becomes quantum ergodic once we add any potential $\binom{Q}{-Q}$, copied across the layers.\footnote{We mention that \cite[Th. 2.2]{KorSa} actually prove that the spectrum becomes purely ac if we add a potential $tQ(v_1)<\dots<tQ(v_\nu)$ with $t$ sufficiently large, assuming each $(H(\theta_\fb))_{jj}$ is nonconstant in $\theta$. In this case the spectral bands are moreover shown to be disjoint.}
\end{rem}

\subsection{Stronger statements}
The following two paragraphs illustrate that one cannot obtain much stronger results than the ones we provide.
\subsubsection{Quantum Unique Ergodicity}\label{sec:que}
In \cite{ALM}, it was suggested to check whether
\begin{equation}\label{e:que}
\lim_{N\to\infty} \sup_{1\le j\le |\Gamma_N|}|\langle \psi_j^{(N)},a_N \psi_j^{(N)}\rangle - \langle a_N\rangle|=0
\end{equation}
as an indication of \emph{quantum unique ergodicity (QUE)}. This would mean that we can avoid the Ces\`aro average in \eqref{e:mainnu1}. This criterion is too strong however, at least in our context, in fact it is already violated for the adjacency matrix on $\Z^d$. See \S\ref{sec:quecor}.

\subsubsection{Eigenvector correlators}\label{sec:eic}
In \cite{Ana}, instead of taking observables $a_N(n)$ which are functions on $\Gamma_N$, a quantum ergodicity theorem was proved more generally for band matrix observables, that is, $K_N(n,m)$, where $K_N(n,m)=0$ if $d(n,m)>R$. It was shown (in Ces\`aro sense) that $\langle \psi_j^{(N)},K_N\psi_j^{(N)}\rangle \approx \langle K_N\rangle_{\lambda_j}$, where $\langle K_N\rangle_{\lambda} = \frac{1}{|\Gamma_N|}\sum_{n,m} K_N(n,m) \Phi_{\lambda}(d(x,y))$ and $\Phi_{\lambda}$ is the spherical function of the tree; it has an explicit form in terms of Chebyshev polynomials. Since $\langle \psi_j^{(N)},K_N\psi_j^{(N)}\rangle = \sum_{n,m}K_N(n,m)\overline{\psi_j^{(N)}(n)}\psi_j^{(N)}(m)$, this shows that the eigenfunction correlator $\overline{\psi_j^{(N)}(n)}\psi_j^{(N)}(m) \approx \frac{1}{|\Gamma_N|} \Phi_{\lambda}(d(n,m))$, a universal quantity; this generalizes the statement that $|\psi_j^{(N)}(n)|^2 \approx \frac{1}{|\Gamma_N|}$.

This stronger statement fails in our case; $\overline{\psi_j^{(N)}(n)}\psi_j^{(N)}(m)$ is not universal, it depends on the basis, even for $\cA_{\Z^d}$. See \S\ref{sec:quecor}.

Still, our proof can be generalized to matrix observables $K_N$. If $\nu=1$, we show that
\[
\frac{1}{N^d}\sum_{j\in \LL_N^d} \big|\langle \psi_j^{(N)},K\psi_j^{(N)}\rangle - \langle K\rangle_{\psi_j^{(N)}}\big|^2\to 0\,,
\]
where $\langle K\rangle_\psi = \frac{1}{N^d}\sum_{n\in\LL_N^d}\sum_{|\tau|\le R} K(n_\fa,n_\fa+\tau_\fa) \langle\psi,\psi(\cdot+\tau_\fa)\rangle$, and $R$ is the width of the band matrix. So in a weak sense, $\overline{\psi_j^{(N)}(n_\fa)}\psi_j^{(N)}(n_\fa+\tau_\fa) \approx \frac{1}{N^d} \langle \psi_j^{(N)},\psi_j^{(N)}(\cdot+\tau_\fa)\rangle$.

\subsection{Structure of the paper}
We prove the general Theorem \ref{thm:maingen} in Section~\ref{sec:proof}. In \S\ref{sec:scalfib}, we prove that \eqref{e:eigenass} is satisfied for $\nu=1$, thereby proving Theorem \ref{thm:mainnu1}. We then discuss Cartesian products in \S\ref{sec:cartesian} and prove Proposition~\ref{prp:cart}. In \S\ref{sec:deco} and \S\ref{sec:tenstro}, we discuss graph decorations, tensor products and strong products of graphs, giving examples of graphs violating quantum ergodicity. In Section~\ref{sec:examp}, we give more specific examples satisfying quantum ergodicity. Finally in Section~\ref{sec:comp}, we discuss complementary results such as quantum unique ergodicity, eigenvector correlators, the Bloch theorem, as well as further criteria for checking \eqref{e:eigenass} based on Bloch varieties considerations.

\section{Proof of the general criterion}\label{sec:proof}

Here we prove Theorem~\ref{thm:maingen}. The argument is very different than the proof for trees \cite{Ana,ALM,AS}. We will use some ideas from \cite{Klimek} where ergodic averages for the continuous Laplacian $-\Delta$ on the torus $\R^d/\Z^d$ are studied, in the high frequency limit.

\subsection{Step 1}
Since $\ee^{-\ii tH_N}\psi_u^{(N)} = \ee^{-\ii t \lambda_u^{(N)}}\psi_u^{(N)}$, $\langle \psi_u^{(N)},\ee^{\ii t H_N}a\ee^{-\ii t H_N}\psi_u^{(N)}\rangle = \langle \psi_u^{(N)},a\psi_u^{(N)}\rangle$ and we have
\begin{equation}\label{e:fisvar}
\langle \psi_u^{(N)},a\psi_u^{(N)}\rangle = \Big\langle\psi_u^{(N)},\frac{1}{T}\int_0^T\ee^{\ii t H_N} a\ee^{-\ii t H_N}\,\dd t\psi_u^{(N)}\Big\rangle \,.
\end{equation}

In the spirit of Egorov's theorem, we show the sandwich $\ee^{\ii t H_N} a\ee^{-\ii t H_N}$ can be expressed as a kind of phase space operator.

Let $\LL_N^d = [\![0,N-1]\!]^d$ and define $U:\ell^2(\Gamma_N)\to \mathop\oplus_{j\in \LL_N^d}\ell^2(V_f)$ by
\begin{equation}\label{e:u}
(U\psi)_j(v_n) = \frac{1}{N^{d/2}} \sum_{k\in \LL_N^d}\ee^{-\frac{\ii j_{\fb}}{N}\cdot k_{\fa}} \psi(v_n+k_{\fa})\,.
\end{equation}
We see as in \cite[\S 3.2]{BoSa} that $U$ is unitary and
\begin{equation}\label{e:uhu-}
UH_NU^{-1} = \mathop\oplus_{j\in\LL_N^d} H\Big(\frac{j_{\fb}}{N}\Big) \,,
\end{equation}
where $H(\theta_\fb)$ acts on $\ell^2(V_f)$ by
\begin{equation}\label{e:htheta}
H(\theta_\fb)f(v_n) = \sum_{u\sim v_n} \ee^{\ii\theta_\fb\cdot \lfloor u \rfloor_{\fa}} f(\{u\}_{\fa}) + Q(v_n)f(v_n)\,.
\end{equation}
The sum is over the nearest neighbors $u$ of $v_n$ in the whole graph $\Gamma$ (not just $V_f$) and we have $u = \lfloor u\rfloor_{\fa} + \{u\}_{\fa}$, with $\lfloor u\rfloor_{\fa} \in \Z_{\fa}^d$ and $\{u\}_{\fa}\in V_f$, cf. \eqref{e:sumlat}.

We have $U^{-1}\big((g_j)_{j\in \LL_N^d}\big)= \psi$, where $\psi(k_{\fa}+v_n) = \frac{1}{N^{d/2}}\sum_{r\in \LL_N^d} g_r(v_n)\ee^{\ii k_{\fa}\cdot\frac{r_{\fb}}{N}}$. To see this, note that $(\frac{1}{N^{d/2}} \ee^{-\ii k_{\fa}\cdot j_{\fb}/N})_{j\in\LL_N^d}$ is an orthonormal basis of $\ell^2(\LL_N^d)$ (use the relation $k_{\fa}\cdot j_{\fb} = 2\pi k\cdot j$). So for such $\psi$ we have $(U\psi)_j(v_n) = \frac{1}{N^d}\sum_{k,r\in\LL_N^d} \ee^{-\frac{\ii j_{\fb}}{N}\cdot k_{\fa}} g_r(v_n)\ee^{\ii k_{\fa}\cdot\frac{r_{\fb}}{N}} = \sum_{k\in\LL_N^d} \langle \frac{1}{N^{d/2}} \ee^{\frac{-2\pi \ii k}{N}\cdot\bullet },g_{\bullet}(v_n)\rangle_{\ell^2(\LL_N^d)}(\frac{1}{N^{d/2}}\ee^{-\frac{\ii j_{\fb}}{N}\cdot k_{\fa}}) = g_j(v_n)$.

Note that $H(\theta_\fb)$ is a $\nu\times\nu$ matrix with eigenvalues $E_1(\theta_\fb),\dots,E_\nu(\theta_\fb)$. Let $P_s(\theta_\fb)$ be the corresponding eigenprojections.

Let $e_r^{(N)}(k):=\frac{1}{N^{d/2}}\ee^{2\pi\ii k\cdot r/N}$. Given $F\in \ell^2(\LL_N^{2d}\times V_f^2)$, we now let
\begin{equation}\label{e:opnper}
\opn(F)\psi(k_{\fa}+v_n): = \sum_{r\in\LL_N^d}\sum_{\ell=1}^\nu (U\psi)_r(v_\ell)F(k,r;v_n,v_\ell) e_r^{(N)}(k)\,,
\end{equation}

The ``quantization'' \eqref{e:opnper} is such that if $F(k,r;v_n,v_\ell) = F(k_\fa+v_n)\delta_{v_n,v_\ell}$, then $\opn(F)\psi = F\psi$. The presence of $\delta_{v_n,v_\ell}$ may seem unusual; indeed it would not be here if we were dealing with just the adjacency matrix on $\Z^d$. The presence of $\delta_{v_n,v_\ell}$ is related to the fact that the Floquet transform \eqref{e:u} is only a partial transform in the sense that it keeps $v_n$ fixed.

Define 
\begin{multline}\label{e:ft}
F_T(k,r;v_n,v_\ell): = \sum_{m\in \LL_N^d}\sum_{q,s,w=1}^\nu  \frac{1}{T}\int_0^T\ee^{\ii t[E_s(\frac{r_{\fb}+m_{\fb}}{N})-E_w(\frac{r_{\fb}}{N})]}\,\dd t \, \\
\times P_s\Big(\frac{r_{\fb}+m_{\fb}}{N}\Big)(v_n,v_q)a_m^{(N)}(v_q)P_w\Big(\frac{r_{\fb}}{N}\Big)(v_q,v_\ell)e_{m}^{(N)}(k)\,,
\end{multline}
where $a_m^{(N)}(v_q):= \langle \frac{\ee^{\frac{\ii m_{\fb}\cdot \bullet_\fa}{N}}}{N^{d/2}},a(v_q + \bullet_\fa)\rangle_{\ell^2(\LL_N^d)}$ are the Fourier coefficients of $a$.

\begin{lem}\label{lem:egorov}
We have 
\[
\frac{1}{T}\int_0^T \ee^{\ii tH_N}a\ee^{-\ii tH_N}\,\dd t = \opn(F_T).
\]
\end{lem}
Although the definitions are somewhat long, the meaning is straightforward: this sandwich can be expressed in phase space. $F_T$ ``smoothes'' the function over different eigenvalues of the phase space operator, and $\opn$ gives the averaging under which this occurs. 
\begin{proof} 
First, we expand $\psi$ in order to relate it to the form of $\opn(F_T)$. 
\[
\psi(k_{\fa}+v_n) = (U^{-1}U\psi)(k_{\fa}+v_n) = \sum_{r\in \LL_N^d} (U\psi)_r(v_n)e_r^{(N)}(k) \,.
\]
Recalling \eqref{e:uhu-}, we obtain
\[
(H_N\psi)(k_{\fa}+v_n) = \sum_{r\in \LL_N^d} \Big[H\Big(\frac{r_{\fb}}{N}\Big)(U\psi)_r\Big](v_n) e_r^{(N)}(k) \,.
\]
Knowing this, we can now examine the operator $\ee^{\ii tH_N}a\ee^{-\ii tH_N}$ and expand over the various $(v_n,v_q)$. This yields
\[
(\ee^{\ii tH_N}a\ee^{-\ii tH_N}\psi)(k_{\fa}+v_n) = \sum_{r\in \LL_N^d}\sum_{q=1}^\nu \ee^{\ii tH(\frac{r_{\fb}}{N})}(v_n,v_q)(Ua\ee^{-\ii tH_N}\psi)_r(v_q)e_r^{(N)}(k)\,.
\]
Expanding $a(v_q+n_{\fa}) = \frac{1}{N^{d/2}} \sum_{m\in \LL_N^d} a_m^{(N)}(v_q)\ee^{\frac{\ii m_{\fb}\cdot n_{\fa}}{N}}$, we have
\begin{align*}
(Ua\ee^{-\ii tH_N}\psi)_r(v_q) &= \frac{1}{N^{d}}\sum_{n\in\LL_N^d}\sum_{m\in \LL_N^d} \ee^{\frac{-\ii r_{\fb}}{N}\cdot  n_{\fa}}a_m^{(N)}(v_q) \ee^{\frac{\ii m_{\fb}\cdot n_{\fa}}{N}}(\ee^{-\ii tH_N}\psi)(v_q+n_{\fa}) \\
&= \frac{1}{N^{d/2}} \sum_{m\in \LL_N^d} a_m^{(N)}(v_q) (U\ee^{-\ii tH_N}\psi)_{r-m}(v_q)\,.
\end{align*}
Here, $r-m$ is understood in $(\Z/N\Z)^d$. More precisely, if $r_i-m_i$ is negative for some $i$, it is replaced by $N+r_i-m_i$. But this last term can be further reduced as
\[
(U\ee^{-\ii tH_N}\psi)_{r-m}(v_q) = [\ee^{-\ii tH(\frac{r_{\fb}-m_{\fb}}{N})}(U\psi)_{r-m}](v_q) = \sum_{\ell=1}^\nu \ee^{-\ii tH(\frac{r_{\fb}-m_{\fb}}{N})}(v_q,v_\ell)(U\psi)_{r-m}(v_\ell) \,.
\]
Moreover, we can write $H(\theta_\fb) = \sum_{s=1}^\nu E_s(\theta_\fb) P_s(\theta_\fb)$ through its eigendecomposition. Similarly, $\ee^{\pm \ii tH(\theta_\fb)} = \sum_{s=1}^\nu \ee^{\pm \ii t E_s(\theta_\fb)}P_s(\theta_\fb)$. Applying this gives us
\begin{multline*}
(\ee^{\ii tH_N}a\ee^{-\ii tH_N}\psi)(k_{\fa}+v_n) = \frac{1}{N^{d/2}}\sum_{r,m\in \LL_N^d} \sum_{q,\ell,s,w=1}^\nu \ee^{\ii t E_s(\frac{r_{\fb}}{N})}P_s\Big(\frac{r_{\fb}}{N}\Big)(v_n,v_q) \\
\times a_m^{(N)}(v_q) \ee^{-\ii t E_w(\frac{r_{\fb}-m_{\fb}}{N})}P_w\Big(\frac{r_{\fb}-m_{\fb}}{N}\Big)(v_q,v_\ell)(U\psi)_{r-m}(v_\ell) e_r^{(N)}(k)\\
=\frac{1}{N^{d/2}}\sum_{r,m\in \LL_N^d} \sum_{q,\ell,s,w=1}^\nu  \ee^{\ii t[E_s(\frac{r_{\fb}+m_{\fb}}{N})-E_w(\frac{r_{\fb}}{N})]}P_s\Big(\frac{r_{\fb}+m_{\fb}}{N}\Big)(v_n,v_q)\\
\times a_m^{(N)}(v_q) P_w\Big(\frac{r_{\fb}}{N}\Big)(v_q,v_\ell)(U\psi)_r(v_\ell) e_{r+m}^{(N)}(k)
\end{multline*}
with $r+m$ again understood in $(\Z/N\Z)^d$.
Since $\frac{1}{N^{d/2}}e_{r+m}^{(N)}(k) = e_r^{(N)}(k)e_m^{(N)}(k)$, we get
\[
\frac{1}{T}\int_0^T \ee^{\ii tH}a\ee^{-\ii tH}\,\dd t \psi(k_{\fa}+v_n) = \sum_{r\in\LL_N^d}	\sum_{\ell=1}^\nu (U\psi)_r(v_\ell)  F_T(k,r;v_n,v_\ell) e_r^{(N)}(k)\,,
\]
with $F_T$ in \eqref{e:ft}. Therefore, according to \eqref{e:opnper}, $\frac{1}{T}\int_0^T \ee^{\ii tH_N}a\ee^{-\ii tH_N}\,\dd t = \opn(F_T)$. 
\end{proof}

\subsection{Step 2} Now we observe that if $E_s(\frac{r_{\fb}+m_{\fb}}{N})-E_w(\frac{r_{\fb}}{N})\neq 0$ for some $m\in\LL_N^d$ and $s,w\in \{1,\dots,\nu\}$, then the corresponding term in $F_T$ vanishes as $T\to \infty$. So define
\begin{multline}\label{e:blim}
b(k,r,v_n,v_\ell) = \sum_{m\in \LL_N^d}\sum_{q,s,w=1}^\nu \mathbf{1}_{S_r}(m,s,w)P_s\Big(\frac{r_{\fb}+m_{\fb}}{N}\Big)(v_n,v_q) \\
\times a_m^{(N)}(v_q) P_w\Big(\frac{r_{\fb}}{N}\Big)(v_q,v_\ell)e_{m}^{(N)}(k)\,,
\end{multline}
where $S_r = \{(m,s,w):E_s(\frac{r_{\fb}+m_{\fb}}{N})-E_w(\frac{r_{\fb}}{N}) = 0\}$. 

\begin{lem}
We have convergence in norm,\footnote{It is worthwhile to note that in the case of trees \cite{Ana,ALM,AS}, we usually evolve the dynamical system in time $T$, essentially up to the girth of the graph, take the size of the graph $N\to\infty$, then finally take $T\to\infty$. Here we first consider the equilibrium limit in $T$, then take $N\to\infty$ in the end of the proof.}
\[
\lim_{T\rightarrow \infty} \|\opn(F_T)-\opn(b)\|_{HS}^2=0.
\]
\end{lem}

\begin{proof}

We use the special basis $\phi_{r,v_\ell}^{(N)} = e_r^{(N)}\mathop\otimes \delta_{v_\ell}$ of $\ell^2(V_N)$. That is, $\phi_{r,v_\ell}^{(N)}(k_{\fa}+v_q) = e_r^{(N)}(k)\delta_{v_\ell}(v_q) = \frac{\ee^{2\pi\ii r\cdot k/N}}{N^{d/2}}\delta_{v_\ell}(v_q)$. By \eqref{e:u}, $(U\phi_{r,v_\ell})_j(v_q) = \langle e_j^{(N)},e_r^{(N)}\rangle_{\ell^2(\LL_N^d)}\delta_{v_\ell}(v_q) = \delta_{j,r}\delta_{v_\ell}(v_q)$. By definition \eqref{e:opnper}, this implies $\opn(F)\phi_{r,v_\ell}^{(N)}(k_{\fa}+v_n) = F(k,r,v_n,v_\ell)e_r^{(N)}(k)$. 

Note that $\|F(\cdot,r,\star,v_\ell)e_n^{(N)}(\cdot)\|^2_{\ell^2(\Gamma_N)} = \frac{1}{N^d} \|F(\cdot,r,\star,v_\ell)\|^2_{\ell^2(\Gamma_N)}$, where $\cdot$ runs over $k\in\LL_N^d$ and $\star$ runs over $v_n\in V_f$. Therefore, 
\begin{equation}\label{eq:opnred}
\|\opn(F)\|_{HS}^2 = \sum_{r\in\LL_N^d}\sum_{\ell=1}^\nu \|\opn(F)\phi_{r,v_\ell}^{(N)}\|^2 = \frac{1}{N^d}\sum_{r\in\LL_N^d}\sum_{\ell=1}^\nu \|F(\cdot,r,\star,v_\ell)\|^2    
\end{equation}
To prove the lemma, we should thus examine the norm of the symbols,
\begin{multline*}
\big\|F_T(\cdot,r,\star,v_\ell)-b(\cdot,r,\star,v_\ell)\big\|^2 = \bigg\| \sum_{m\in\LL_N^d}\sum_{q,s,w=1}^\nu \mathbf{1}_{S^c_r}(m,s,w) \frac{\ee^{\ii T[E_{s}(\frac{r_{\fb}+m_{\fb}}{N})-E_w(\frac{r_{\fb}}{N})]}-1}{T[E_{s}(\frac{r_{\fb}+m_{\fb}}{N})-E_w(\frac{r_{\fb}}{N})]}\\
\times P_s\Big(\frac{r_{\fb}+m_{\fb}}{N}\Big)(\star,v_q)a_m^{(N)}(v_q)P_w\Big(\frac{r_{\fb}}{N}\Big)(v_q,v_\ell)e_m^{(N)}(\cdot)\bigg\|^2\,.
\end{multline*}
This implies that
\begin{multline*}
\|\opn(F_T)-\opn(b)\|_{HS}^2= \frac{1}{T^2N^d}\sum_{r,m\in \LL_N^d}\sum_{\ell=1}^\nu \bigg\| \sum_{q,s,w=1}^\nu \mathbf{1}_{S^c_r}(m,s,w) \frac{\ee^{\ii T[E_{s}(\frac{r_{\fb}+m_{\fb}}{N})-E_w(\frac{r_{\fb}}{N})]}-1}{E_{s}(\frac{r_{\fb}+m_{\fb}}{N})-E_w(\frac{r_{\fb}}{N})}\\
\times P_s\Big(\frac{r_{\fb}+m_{\fb}}{N}\Big)(\star,v_q)a_m^{(N)}(v_q)P_w\Big(\frac{r_{\fb}}{N}\Big)(v_q,v_\ell)\bigg\|^2_{\C^\nu} \le \frac{C_{N,a}}{T^2}\,,
\end{multline*}
where $C_{N,a}$ is finite for any $N$ and is independent of $T$. Taking $T\to \infty$ yields that $\opn(F_T)\to \opn(b)$ in HS norm.
\end{proof}

\subsection{Step 3}
We are thus reduced to studying $\opn(b)$ with $b$ given in \eqref{e:blim}.

Note that $\sum_{p=1}^\nu P_p(\theta) =\mathrm{id}$, so $\sum_{p=1}^\nu P_p(\theta)(v_i,v_j) = \delta_{v_i,v_j}$. Therefore, if we remove the $\mathbf{1}_{S_r}$ term, \eqref{e:blim} becomes
\[
\sum_{m\in\LL_N^d} a_m^{(N)}(v_n) e_m^{(N)}(k) \delta_{v_n,v_\ell} = a(k_{\fa}+v_n)\delta_{v_n,v_\ell}
\]
and the corresponding $\opn$ applied to $\psi$ simply gives $a(k_{\fa}+v_n)\psi(k_{\fa}+v_n)$. Hence, $\opn(b)\psi$ is just $a\psi$ but with many suppressed Floquet modes.

Let $\overline{a}$ be the part of $b$ corresponding to $m=0$.
Let $\tilde{a} = a - \opn(\overline{a})$ and $c = b - \overline{a}$. Then collecting the previous steps, we have
\begin{multline*}
\sum_{u\in\Gamma_N} |\langle \psi_u^{(N)}, \tilde{a}\psi_u^{(N)}\rangle|^2 = \sum_{u\in\Gamma_N} \lim_{T\to\infty} |\langle \psi_u^{(N)}, \opn(F_T- \overline{a})\psi_u^{(N)}\rangle|^2 \\
\le \sum_{u\in \Gamma_N}\lim_{T\to\infty} 2(\|\opn(c)\psi_u^{(N)}\|^2 + \|\opn(F_T - b)\psi_u^{(N)}\|^2) = 2\|\opn(c)\|^2_{HS}\,.
\end{multline*}

\begin{proof}[Proof of \eqref{e:main}]
It now suffices to show that $\lim\limits_{N\to\infty}\frac{1}{|\Gamma_N|}\|\opn(c)\|_{HS}^2 = 0$. Using \eqref{eq:opnred}, we have $\frac{1}{|\Gamma_N|}\|\opn(c)\|_{HS}^2 =\frac{1}{\nu N^{2d}}\sum_{r\in \LL_N^d}\sum_{\ell=1}^\nu \|c(\cdot,r,\star,v_\ell)\|_{\ell^2(V_N)}^2$.

We thus consider
\begin{multline*}
\frac{1}{N^{2d}} \sum_{r\in\LL_N^d}\sum_{\ell=1}^\nu \bigg\|\sum_{m\neq 0}\sum_{q,s,w=1}^\nu \mathbf{1}_{S_r}(m,s,w)P_s\Big(\frac{r_{\fb}+m_{\fb}}{N}\Big)(\star,v_q) a_m^{(N)}(v_q) P_w\Big(\frac{r_{\fb}}{N}\Big)(v_q,v_\ell)e_{m}^{(N)}(\cdot)\bigg\|^2\\
= \frac{1}{N^{2d}} \sum_{r\in\LL_N^d}\sum_{\ell=1}^\nu \sum_{m\neq 0} \sum_{n=1}^\nu\Big|\sum_{q,s,w=1}^\nu \mathbf{1}_{S_r}(m,s,w)P_s\Big(\frac{r_{\fb}+m_{\fb}}{N}\Big)(v_n,v_q) a_m^{(N)}(v_q) P_w\Big(\frac{r_{\fb}}{N}\Big)(v_q,v_\ell)\Big|^2
\end{multline*}

Denote $P_s:=P_s(\frac{r_{\fb}+m_{\fb}}{N})$, $P_w:=P_w(\frac{r_{\fb}}{N})$ and expand the square modulus to get
\begin{multline}\label{e:expanc}
\frac{1}{N^{2d}} \sum_{r\in\LL_N^d}\sum_{\ell=1}^\nu \sum_{m\neq 0} \sum_{n=1}^\nu \sum_{q,s,w,q',s',w'=1}^\nu \mathbf{1}_{S_r}(m,s,w)P_s(v_n,v_q) a_m^{(N)}(v_q) P_w(v_q,v_\ell)\\
\times \mathbf{1}_{S_r}(m,s',w')\overline{P_{s'}(v_n,v_{q'}) a_m^{(N)}(v_{q'}) P_{w'}(v_{q'},v_\ell)}\,.
\end{multline}

But
\[
\sum_{n=1}^\nu P_s(v_n,v_q)\overline{P_{s'}(v_n,v_{q'})} = \sum_{n=1}^\nu (P_s\delta_{v_q})(v_n)\overline{(P_{s'}\delta_{v_{q'}})(v_n)} = \langle P_{s'}\delta_{v_{q'}},P_s\delta_{v_q}\rangle \,.
\]
Similarly, $\sum_{\ell=1}^\nu P_w(v_q,v_\ell)\overline{P_{w'}(v_{q'},v_\ell)}  = \langle P_w\delta_{v_q},P_{w'}\delta_{v_{q'}}\rangle$. If $E_s\neq E_{s'}$ or $E_w\neq E_{w'}$, these scalar products vanish. So \eqref{e:expanc} is concentrated on the $s',w'$ for which $E_{s'}=E_s$ and $E_{w'}=E_w$, in which case $\mathbf{1}_{S_r}(m,s,w) = \mathbf{1}_{S_r}(m,s',w')$ and we obtain
\begin{multline*}
\frac{1}{N^{2d}} \sum_{r\in\LL_N^d} \sum_{m\neq 0} \sum_{q,s,w,q',s',w'=1}^\nu \mathbf{1}_{S_r}(m,s,w)\langle P_{s'}\delta_{v_{q'}},P_s\delta_{v_q}\rangle a_m^{(N)}(v_q) \langle P_w\delta_{v_q},P_{w'}\delta_{v_{q'}}\rangle\overline{a_m^{(N)}(v_{q'})}\\
=\frac{1}{N^{2d}}\sum_{\substack{m\in\LL_N^d\\m\neq 0}}\sum_{q,q'=1}^\nu \overline{a_m^{(N)}(v_{q'})}a_m^{(N)}(v_q)\sum_{r\in\LL_N^d}\sum_{s,w,s',w'=1}^\nu \mathbf{1}_{A_m}(r,s,w)\langle P_{s'}\delta_{v_{q'}},P_s\delta_{v_q}\rangle \langle P_w\delta_{v_q},P_{w'}\delta_{v_{q'}}\rangle\,,
\end{multline*}
where $A_m = \{(r,s,w):E_s(\frac{r_{\fb}+m_{\fb}}{N})-E_w(\frac{r_{\fb}}{N})=0\}$ and we used that $(m,s,w)\in S_r\iff (r,s,w)\in A_m$. By hypothesis \eqref{e:eigenass}, we know that
\begin{equation}\label{e:pervar}
\lim_{N\rightarrow\infty} \sup_{\substack{m\in\LL_N^d\\m\neq 0}}\frac{|A_m|}{N^d} =0.
\end{equation}
Since $|\langle P_{s'}\delta_{v_{q'}},P_s\delta_{v_q}\rangle|\le 1$, it follows that the above is
\[
o_N(1) \frac{1}{N^d} \sum_m\sum_{q,q'=1}^\nu \overline{a_m^{(N)}(v_{q'})}a_m^{(N)}(v_q)=o_N(1)\frac{1}{N^d}\sum_{q,q'=1}^\nu \langle a(\cdot_\fa+v_{q'}),a(\cdot_\fa+v_q)\rangle_{\ell^2(\LL_N^d)} = o_N(1)
\]
using $|a(n_\fa+v_q)|\le 1$. This completes the proof of \eqref{e:main}.
\end{proof}
\subsection{Step 4}\label{sec:step4}
Let us now explore the main term $\overline{a}$. Recall that it corresponds to $m=0$ in \eqref{e:blim}. Having $(0,s,w)\in S_r$ means that $E_s(\frac{r_{\fb}}{N}) = E_w(\frac{r_{\fb}}{N})$. This is automatically true for $w=s$. Thus,
\begin{align}\label{e:abar}
\overline{a}&=\sum_{q,s=1}^\nu a_0^{(N)}(v_q)e_{0}^{(N)}(k) P_s\Big(\frac{r_{\fb}}{N}\Big)(v_n,v_q) \bigg(P_s\Big(\frac{r_{\fb}}{N}\Big)(v_q,v_\ell) + \sum_{\substack{w\neq s\\ E_s=E_w}} P_w\Big(\frac{r_\fb}{N}\Big)(v_q,v_\ell)\bigg) \nonumber\\
&= \sum_{q=1}^\nu \langle a(\cdot+v_q)\rangle \sum_{s=1}^{\nu'}P_{E_s}\Big(\frac{r_{\fb}}{N}\Big)(v_n,v_q) P_{E_s}\Big(\frac{r_{\fb}}{N}\Big)(v_q,v_\ell)  \,,
\end{align}
where $\langle a(\cdot+v_q)\rangle = \frac{1}{N^d}\sum_{n\in \LL_N^d}a(n_{\fa}+v_q)$, $\nu'\le \nu$ is the number of \emph{distinct} eigenvalues of $H(\theta_\fb)$ and $P_{E_s}(\theta_\fb) = \sum_{E_w=E_s} P_w(\theta_\fb)$ is the orthogonal projection onto the eigenspace corresponding to $E_s(\theta_\fb)$. In general, $\nu'$ is independent of $\theta_\fb$, i.e. the multiplicity of $E_s(\theta_\fb)$ is fixed, expect perhaps for finitely many exceptional $\theta_\fb$ (see e.g. \cite[Chapter II.1.1]{Kato}).

\begin{proof}[Proof of \eqref{e:asavboun}-\eqref{e:avconverge}]
By the definition of $\opn$, we can write out 
\begin{align*}
\langle \psi,\opn(\overline{a})\psi\rangle &= \sum_{k\in\LL_N^d}\sum_{v_n\in V_f} \overline{\psi(k_{\fa}+v_n)}[\opn(\overline{a})\psi](k_{\fa}+v_n) \\
&= \sum_{v_n\in V_f}\sum_{r\in\LL_N^d}\sum_{\ell=1}^\nu (U\psi)_r(v_\ell) \overline{a}(r,v_n,v_\ell) \sum_{k\in \LL_N^d} \overline{\psi(k_{\fa}+v_n)} e_r^{(N)}(k)\,.
\end{align*}

But $\sum_k \overline{\psi(k_{\fa}+v_n)} e_r^{(N)}(k) = \overline{(U\psi)_r(v_n)}$. Thus,
\begin{align*}
\langle \psi, \opn(\overline{a})\psi\rangle &= \sum_{v_n\in V_f}\sum_{r\in\LL_N^d}\sum_{\ell=1}^\nu (U\psi)_r(v_\ell)\overline{(U\psi)_r(v_n)} \overline{a}(r,v_n,v_\ell)\\
&= \sum_{q=1}^\nu \sum_{r\in\LL_N^d}\sum_{\ell=1}^\nu\sum_{s=1}^{\nu'} P_{E_s}(v_q,v_\ell) (U\psi)_r(v_\ell)\sum_{n=1}^\nu \overline{P_{E_s}(v_q,v_n)(U\psi)_r(v_n)}\langle a(\cdot+v_q)\rangle \\
&= \sum_{q=1}^\nu \langle a(\cdot+v_q)\rangle \sum_{r\in \LL_N^d} \sum_{s=1}^{\nu'}[P_{E_s}(U\psi)_r](v_q) \overline{[P_{E_s}(U\psi)_r](v_q)} \,.
\end{align*}
where $P_s = P_s(\frac{r_{\fb}}{N})$. We have shown that
\begin{equation}\label{e:opnabarpsi}
\langle \psi, \opn(\overline{a})\psi\rangle = \sum_{q=1}^\nu \langle a(\cdot + v_q)\rangle \sum_{r\in \LL_N^d} \sum_{s=1}^{\nu'}\Big| \Big[P_{E_s}\Big(\frac{r_{\fb}}{N}\Big)(U\psi)_r\Big](v_q)\Big|^2 \,.
\end{equation}

In the special case where $\langle a(\cdot+v_q)\rangle = \langle a(\cdot +v_1)\rangle$ for $q=1,\dots,\nu$, the above reduces to
\[
\langle a(\cdot +v_1)\rangle  \sum_{r\in \LL_N^d} \sum_{s=1}^{\nu'} \|P_{E_s}(U\psi)_r\|_{\C^\nu}^2 = \langle a(\cdot + v_1)\rangle \sum_{r\in \LL_N^d} \|(U\psi)_r\|_{\C^\nu}^2 = \langle a(\cdot+v_1)\rangle \|\psi\|^2 \,.
\]
In particular, $\psi = \psi_u^{(N)}$ gives the uniform average $\langle a(\cdot+v_1)\rangle = \frac{1}{N^d}\sum_{n\in \LL_N^d}a(n_\fa+v_1) = \frac{1}{|\Gamma_N|}\sum_{v\in \Gamma_N} a(v)$. This proves \eqref{e:avconverge}.

In the same way, if $a$ is real-valued, we deduce \eqref{e:asavboun} from \eqref{e:opnabarpsi}.

We finally show that if $a$ is locally constant, then it suffices to ask that \eqref{e:eigenass} holds for $w=s$, $1\le s\le \nu$. In fact, going back to the calculation in Step 3 preceding \eqref{e:pervar}, we note that $\sum_{s'}\langle P_{s'}\delta_{v_{q'}},P_s\delta_{v_q}\rangle = \langle P_{E_s}\delta_{v_{q'}},P_s\delta_{v_q}\rangle = P_s(v_{q'},v_q)$. Similarly $\sum_{w'} \langle P_w \delta_{v_q},P_{w'}\delta_{v_{q'}}\rangle = P_w(v_q,v_{q'})$.

Since $\overline{a_m^{(N)}(v_{q'})}a_m^{(N)}(v_q)=|a_m^{(N)}(v_1)|^2 =:|a_m^{(N)}|^2$ is independent of $q,q'$, we can now sum the kernels, $\sum_{q,q'} P_s(v_{q'},v_q)P_w(v_q,v_{q'}) = \sum_q (P_wP_s)(v_q,v_q) = \delta_{s,w}$. We thus see that when $a$ is locally constant, then $\frac{1}{N^d}\|\opn(c)\|_{HS}^2$ reduces to
\begin{equation}\label{e:cforlocalcon}
\frac{1}{N^{2d}}\sum_{\substack{m\in\LL_N^d\\m\neq 0}}|a_m^{(N)}|^2\sum_{r\in\LL_N^d}\sum_{s=1}^\nu \mathbf{1}_{A_m}(r,s,s) .
\end{equation}
If \eqref{e:eigenass} holds for $w=s$, $1\le s\le \nu$, the above is $o_N(1)\frac{1}{N^d}\|a\|^2 = o_N(1)$ as required.
\end{proof}

\begin{rem}[Necessity of the assumptions]\label{rem:nec}
In the previous proof, the only inequality that we used is in Step 3, when bounding the variance by the Hilbert-Schmidt norm of the evolved observable. This bound is standard in proofs of quantum ergodicity, it seems unlikely that we can avoid it. On the other hand, the decay of the Hilbert-Schmidt norm almost necessitates \eqref{e:eigenass}. In fact, if in Step 3 we consider, for fixed $\widehat{m}\neq 0$, the observable $a(k_\fa+v_q) = \ee^{\frac{\ii \widehat{m}_\fb\cdot k_\fa}{N}}\otimes \mathbf{1}_\nu(v_q)$, which is locally constant, then $a_m^{(N)} = N^{d/2}$ if $m=\widehat{m}$ and zero otherwise, so \eqref{e:cforlocalcon} reduces to $\frac{\#\{(r,s):E_s(\frac{r_\fb+\widehat{m}_\fb}{N})-E_s(\frac{r_\fb}{N})=0\}}{N^d}$. This quantity therefore has to be $o_N(1)$ for every $\widehat{m}\neq 0$ if we want the HS norm to go to zero $\forall a$. For $w\neq s$ we need an assumption. Namely, if we can choose the normalized eigenvectors $f_s^{\theta_\fb}$ corresponding to $E_s(\theta_\fb)$ such that for some $v_\ell$, we have $f_s^{\theta_\fb}(v_\ell)\neq 0$ for all $s,\theta$, then by taking $a(k_\fa+v_q) = \ee^{\frac{\ii \widehat{m}_\fb\cdot k_\fa}{N}} \delta_{v_\ell}(v_q)$, we see that $\frac{|A_m|}{N^d}=o_N(1)$ must hold $\forall m\neq 0$ for the HS norm to go to zero (argue similarly in the calculation preceding \eqref{e:pervar}).
\end{rem}

\begin{rem}\label{rem:I}
The main theorem holds more generally if instead of summing over the whole spectrum in \eqref{e:main}, we sum over eigenvalues in some interval $I$, in which case we only need \eqref{e:eigenass} to hold on $I$. To see this, we slightly modify the proof as follows~: in \eqref{e:fisvar}, we insert a spectral projection $\chi_I(H_N)$, so the operator is now $\frac{1}{T}\int_0^T \ee^{\ii tH_N}a\ee^{-\ii tH_N}\chi_I(H_N)\,\dd t$. In \eqref{e:ft}, we replace the sum over all $w$ by the sum over $E_w(\frac{r_\fb}{N})\in I$. In fact, by adding the spectral projection through the proof of Lemma~\ref{lem:egorov}, we now get $(U\ee^{-\ii tH}\chi_I(H)\psi)_{r-m} = \ee^{-\ii tH(\frac{r_\fb-m_\fb}{N})}\chi_I(H(\frac{r_\fb-m_\fb}{N}))(U\psi)_{r-m}$. Consequently, the limiting symbol $b$ now also sums over $E_w(\frac{r_\fb}{N})\in I$ instead. The proofs carry over \textit{mutatis mutandis}. 

In the end, the symbol $\overline{a}$ in \eqref{e:abar} now sums over $E_s(\frac{r_\fb}{N})\in I$. This gives the illusory impression that the weighted average changes, which of course makes no sense as the term $\langle \psi_u^{(N)},a\psi_u^{(N)}\rangle$ should approach a fixed quantity whether the Ces\`aro mean is over the whole spectrum or not. However the quantity $\langle \psi_u^{(N)},\opn(\overline{a})\psi_u^{(N)}\rangle$ is indeed the same as before. In fact, if we know that $\lambda_u^{(N)}\in I$, we may again insert a projector so that $(U\psi_u^{(N)})_r$ in \eqref{e:opnabarpsi} becomes $(U\chi_I(H_N)\psi_u^{(N)})_r = \chi_I(H(\frac{r_\fb}{N}))(U\psi_u^{(N)})_r$, so the sum over all $E_s$ in \eqref{e:opnabarpsi} reduces to the sum over $E_s(\frac{r_\fb}{N})\in I$, which is what we obtained when averaging over $I$.
\end{rem}

\section{Special classes of graphs}
In this section we discuss the validity of assumption \eqref{e:eigenass} for various classes of graphs. We start with graphs having $\nu=1$, proving Theorem~\ref{thm:mainnu1}. We then discuss Cartesian products, proving Proposition~\ref{prp:cart}, and conclude with graph decorations, tensor and strong products, proving Propositions~\ref{prp:zno} and \ref{prp:counterass} along the way.

\subsection{Scalar fibers}\label{sec:scalfib}
Step 4 in \S\ref{sec:step4} shows that if $\nu=1$, then $\langle \psi_u^{(N)},\opn(\overline{a})\psi_u^{(N)}\rangle = \langle a\rangle$. To prove Theorem~\ref{thm:mainnu1}, it remains to establish \eqref{e:eigenass} in this context. Here of course $w=s$.

If $\nu=1$, then the graph is $2D$-regular for some $D\in \N$. In fact, $V_f = \{o\}$ for some $o\in \cC_{\mathfrak{a}}$, and $\Gamma = \Z_{\fa}^d + \{o\}$. If $u\sim o$, then $u = \lfloor u\rfloor_{\fa} + o$. By translation invariance we have $u-n_\fa\sim o-n_\fa$. Applying this to $n_\fa = \lfloor u\rfloor_\fa$ gives $o \sim o - \lfloor u\rfloor _\fa$. We may thus arrange the neighbors of $o$ into $\cN_o^+ \cup \cN_o^-$, where $\cN_o^+ = \{o+n_\fa\}$ and $\cN_o^- = \{o-n_\fa\}$, for some $D$ nonzero integers $n_\fa = \sum_{i=1}^d n_i \fa_i$ with $n_i\in \{0,1,\dots\}$. Since the rest of the graph is just a periodic copy of the star around $o$, we see it is $2D$-regular.

If $\nu=1$, then the potential $Q$ must be constant. We assume without loss of generality that $Q=0$.

\begin{proof}[Proof of Theorem~\ref{thm:mainnu1}]
The $\nu\times \nu$ matrix $H(\theta_\fb)$ is now just a scalar given by
\[
H(\theta_\fb) = \sum_{u\sim o} \ee^{\ii\theta_\fb\cdot \lfloor u\rfloor_\fa} = 2\sum_{p=1}^D \cos(2\pi \theta \cdot n^{(p)}) 
\]
for some $n^{(1)},\dots,n^{(D)}\in \{0,1,\dots\}^d\setminus\{0\}$. We only have one eigenvalue here given by $E(\theta_\fb) = H(\theta_\fb)$. So we should show that for any fixed $m\neq 0$, the equation
\begin{equation}\label{eq:tosolve}
E\Big(\frac{r_\fb + m_\fb}{N}\Big) - E\Big(\frac{r_\fb}{N}\Big) = 2\sum_{p=1}^D \Big(\cos\Big(2\pi \frac{(r+m) \cdot n^{(p)}}{N}\Big) - \cos\Big(2\pi \frac{r \cdot n^{(p)}}{N}\Big) \Big) = 0
\end{equation}
has $o(N^d)$ solutions in $r\in\LL_N^d$. By the sum to product formula, we are lead to consider the zeroes of
\begin{equation}\label{e:rooti}
f_m\Big(\frac{r}{N}\Big) := \sum_{p=1}^D \sin \Big(\pi \frac{m\cdot n^{(p)}}{N}\Big) \sin\Big(\pi\frac{(2r+m)\cdot n^{(p)}}{N}\Big) \,.
\end{equation}
For this, we consider the projection of the surface $A_m = \{r\in\LL_N^d :f_m(\frac{r}{N})=0\}$ onto a vector $\phi\in \LL_N^d$ to be specified. More precisely, given $j\in\LL_N^d$, we write $j = r+y \phi$, for $r\in\phi^\bot$ and $y = \frac{\langle \phi,j\rangle}{\|\phi\|^2}$. Note that $y\in [0,N-1]$ since $0\le \sum\phi_i j_i\le (N-1)\sum \phi_i \le (N-1)\sum \phi_i^2$ for $\phi\in\LL_N^d$. We will show that for fixed $r\in \phi^\bot$, there are at most $M$ points $y$ such that $f_m(\frac{j}{N})=0$, with $M$ independent of $N$. By varying $r\in\phi^\bot$, it follows that $|A_m|\le M| \phi^\bot| \le MN^{d-1} = o(N^d)$ as required.

We therefore consider the function
\[
g_{m,r}(x) = f_m\Big(\frac{r}{N}+ x\phi\Big) = 0
\]
for $x\in [0,1)$.  Denote
\begin{equation}\label{e:abc}
\alpha_p = \sin\Big(\pi \frac{m\cdot n^{(p)}}{N}\Big)\,, \qquad \beta_p = \pi \frac{(2r + m)\cdot n^{(p)}}{N}\,, \qquad \gamma_p = 2\phi \cdot n^{(p)}\,.
\end{equation}
Then 
\[
g_{m,r}(x) = \sum_{p=1}^D \alpha_p\sin(\beta_p+\pi \gamma_p x) = \frac{1}{2\ii}\sum_{p=1}^D \alpha_p(\ee^{\ii \beta_p}\ee^{\ii \pi \gamma_p x} - \ee^{-\ii \beta_p}\ee^{-\ii \pi\gamma_px})\,.
\]
Setting $z = \ee^{\ii \pi x}$, this reduces to
\[
\tilde{g}_{m,r}(z) = \sum_{p=1}^D (\rho_p z^{\gamma_p} + \rho_p' z^{-\gamma_p})
\]
for some $\rho_p,\rho_p'\in \C$. By definition \eqref{e:abc}, $\gamma_p$ is an integer. We thus seek the solutions of $\tilde{g}_{m,r}(z)$ on the unit circle. We have $\tilde{g}_{m,r}(z) = 0$ iff $\sum_{p=1}^D (\rho_p z^{\gamma_\star+\gamma_p} + \rho_p'z^{\gamma_\star-\gamma_p}) = 0$, where $\gamma_\star = \max_p\gamma_p$. This is a polynomial in $z$. By the fundamental theorem of algebra, if this polynomial is nontrivial, it has at most $M=2\max_p \gamma_p$ roots. In turn, we have at most $M$ solutions $x_j$ for $g_{m,r}(x)=0$, and the proof of \eqref{e:pervar} is complete (recall the discussion after \eqref{e:rooti}).

So it remains to check the polynomial $z^{\gamma_*}\tilde{g}_{m,r}(z) = \sum_{p=1}^D (\rho_p z^{\gamma_\star+\gamma_p} + \rho_p'z^{\gamma_\star-\gamma_p})$ is nontrivial. For this, we check that
\begin{enumerate}[1.]
\item At least one $\rho_p$ is nonzero. 
\item We can choose $\phi$ such that $\gamma_p\neq  \gamma_{p'}$ for $p\neq p'$. This way, no two terms in the sum have the same power, so no cancellation can occur. Note that no $\gamma_p$ is zero, so no cancellation can occur from $\rho_{p'}=-\rho_p$.
\end{enumerate}

\paragraph{Proof of 1.}
Since $m\neq 0$, we have $m_j\neq 0$ for some $j$. Note that $o+\fa_j\in \Gamma$ by translation invariance. Since $\Gamma$ is connected, some integer combination $o+\sum_{p=1}^D k_p n^{(p)}_\fa$ of the neighbors of $o$ is $o+\fa_j$, where $k_p\in \Z$ is the number of adjacencies of type $n^{(p)}$ traversed on the geodesic from $o$ to $o+\fa_j$. It follows that
\begin{equation}\label{e:sinsum}
\sin\Big(\pi \frac{m}{N}\cdot \sum_{p=1}^D k_p n^{(p)}\Big)  =\sin\Big(\frac{m_\fb}{2N}\cdot \sum_{p=1}^D k_p n_\fa^{(p)}\Big)  = \sin\Big(\frac{m_\fb}{2N}\cdot \fa_j\Big)  =  \sin \frac{\pi m_j}{N} \neq 0\,.
\end{equation}
If we had $\sin (\pi \frac{m\cdot n^{(p)}}{N})=0$ for all $p$, we would have $\frac{m\cdot n^{(p)}}{N}\in \Z$ for all $p$ and thus $\frac{m}{N}\cdot \sum_{p=1}^D k_p n^{(p)}\in \Z$, contradicting \eqref{e:sinsum}. Thus, $\alpha_p \neq 0$ for at least one $p$. This completes the proof.

\paragraph{Proof of 2.}
We need $\phi$ to avoid the subspaces $V_{p,p'}=\{v: v\cdot n^{(p)} = v\cdot n^{(p')}\}$ for all $D(D-1)$ pairs of $p\neq p'$. Each of these is $d-1$ dimensional, since the $n^{(p)}$ are distinct. 

It is not difficult to see that such a $\phi$ exists. However, we give a quite explicit construction below, which in turn gives an explicit bound on $M$. 

Suppose we give a list of $\ell_D=(d-1)D(D-1)+1$ vectors in $\LL_N^d$ such that any $d$ of them forms a basis. Then each of the subspaces $V_{p,p'}$ can only contain at most $d - 1$ of our vectors, therefore
there must be some vector not contained in any of the subspaces and we are done.

A possible list is given by the row vectors
\[
\left[\begin{array}{ccccc}
     1&1&1&\cdots& 1\\
     1&2&2^2&\cdots&2^{d-1} \\
     1&3&3^2&\cdots&3^{d-1}\\
     \vdots&\vdots&\ddots&\vdots\\
     1&\ell_D&\ell_D^2&\cdots &\ell_D^{d-1}
\end{array}\right] .
\]
Indeed, any subset of $d$ of these vectors, say the ones from the $x_1,\dots , x_d$ rows, forms a Vandermonde matrix with determinant $\prod_{i<j}(x_i - x_j)$, which is nonzero, meaning any set of $d$ vectors is linearly independent. This finishes the proof.
\end{proof}
We may obtain an upper bound over $M = 2\max_p \gamma_p$. In fact, the worst case is if the last vector in the list is the first $\phi$ that avoids all $V_{p,p'}$. In this case, $\gamma_p = 2\phi \cdot n^{(p)} \le 2d\ell_D^{d-1}q$, where $q = \max_{i,p} n_i^{(p)}$, so $M\le 4d \ell_D^{d-1}q$.

\subsection{The case of Cartesian products}\label{sec:cartesian}
The Cartesian product $\Gamma \mathop\square G$ of $\Gamma$ and $G$ is the graph with vertex set $V(\Gamma)\times V(G)$, in which $(u,v)\sim (u',v')$ if either
\begin{enumerate}[(i)]
\item ($u=u'$ and $v\sim v'$),
\item or ($u\sim u'$ and $v=v'$). 
\end{enumerate}

For example, to construct $\Z\mathop\square P_2$, where $P_2$ is the $2$-path, replace each vertex of $\Z$ with a $2$-path, and connect edges between matching vertices. The result is an infinite ladder.
\begin{figure}[h!]
\begin{center}
\setlength{\unitlength}{1cm}
\thicklines
\begin{picture}(1,1)(-1,-1)
   \put(-5,0){\line(1,0){10}}
	 \put(-5,-1){\line(1,0){10}}
	 \put(-3,-1){\line(0,1){1}}
	 \put(-1,-1){\line(0,1){1}}
	 \put(1,-1){\line(0,1){1}}
	 \put(3,-1){\line(0,1){1}}
	 \put(-1,-1){\circle*{.2}}
	 \put(-1,0){\circle*{.2}}
	 \put(-3,-1){\circle*{.2}}
	 \put(-3,0){\circle*{.2}}
	 \put(1,-1){\circle*{.2}}
	 \put(1,0){\circle*{.2}}
	 \put(3,-1){\circle*{.2}}
	 \put(3,0){\circle*{.2}}
\end{picture}
\caption{The ladder graph, $\Z\mathop\square P_2$.}\label{fig:lad}
\end{center}
\end{figure}

Similarly, for $\Z\mathop\square C_p$, where $C_p$ is a $p$-cycle, replace each vertex of $\Z$ with a $p$-cycle, and connect edges between matching vertices (Figure~\ref{fig:cyl}).  The graph is $4$-regular, naturally embedded in $\R^3$, and is clearly $\Z$-periodic with fundamental crystal $V_f = C_p$.
We may endow $C_p$ with a potential $Q$ and copy it in each layer. Then $H(\theta_\fb)f(u,v) = 2\cos 2\pi\theta f(u,v) + f(u,v+1) + f(u,v-1) + Q_v f(u,v)$. In other words, $H(\theta_\fb) = \cA_\Z(\theta_\fb)\mathop\otimes I + I\otimes H_{G_F}$. The eigenvalues are thus $\{2\cos 2\pi\theta + \mu_j\}$, where $\{\mu_j\}_{j=1}^p$ are the eigenvalues of the Schr\"odinger operator of the $p$-cycle. These observations are general~:

\begin{lem}
If $\Gamma$ is a periodic graph with $\nu=1$ and $G_F$ is a finite graph endowed a potential $Q$, then $\Gamma \mathop\square G_F$ is a periodic graph with fundamental crystal $V_f = G_F$ and
\begin{equation}\label{e:cartop}
H_{\Gamma\mathop\square G_F}(\theta_\fb) = H_\Gamma(\theta_\fb)\otimes I + I\otimes H_{G_F}\,.
\end{equation}
\end{lem}
\begin{proof}
Replace each $u\in \Gamma$ by a copy of $G_F$. The result has vertex set $V(\Gamma)\times V(G_F)$. According to rules (i)-(ii), we should have $\cA_{\Gamma \mathop\square G_F} = I\otimes \cA_{G_F} + \cA_\Gamma \otimes I$. This means that if we arrange the vertices of $\Gamma \mathop \square G_F$ as successive $G_F$-layers, then a given $(u,v)$ is connected on the one hand to the neighbors $(u,v')$ in the same layer (rule (i)) and to the neighbors $(u',v)$ outside (rule (ii)). This means that the edges are precisely the old edges of $G_F$ in each layer, as well as bridges between successive layers between the matching vertices. Recalling \eqref{e:htheta}, we see that the $\theta$-dependence only arises in the bridges from $(u,v)$ to another layer (the neighbors within $G_F$ have $\lfloor u\rfloor_{\fa}=0$). The bridges occur precisely at the bridges from $u$ to its neighbors in $\Gamma$. We conclude that $\cA_{\Gamma\mathop\square G_F}(\theta_\fb) = \cA_\Gamma(\theta_\fb)\otimes I + I\otimes \cA_{G_F}$. If we finally endow $G_F$ a potential and copy it across the layers, then $(Qf)(u,v)=Q_vf(u,v)$, so we obtain \eqref{e:cartop} (note that $\cA_{\Gamma}(\theta_\fb)=H_\Gamma(\theta_\fb)$ as $\nu=1$).
\end{proof}

\begin{proof}[Proof of Proposition~\ref{prp:cart}]
Since $\nu=1$ for $\Gamma$, $H_{\Gamma}(\theta_\fb)$ has just one eigenvalue $E_\Gamma(\theta_\fb)$. So the spectrum of $H_{\Gamma\mathop\square G_F}$ is the set $\{E_\Gamma(\theta_\fb) + \mu_j\}$, where $\mu_j$ are the eigenvalues of $H_{G_F}$ on the finite graph $G_F$.

Given nonzero $m$, we should thus control the quantity
\[
E_s\Big(\frac{r_\fb+m_\fb}{N}\Big) - E_w\Big(\frac{r_\fb}{N}\Big) = E_\Gamma\Big(\frac{r_\fb+m_\fb}{N}\Big) - E_\Gamma\Big(\frac{r_\fb}{N}\Big) +\mu_s - \mu_w \,.
\]

Here, $E_\Gamma(\theta_\fb) = 2\sum_{p=1}^D \cos(2\pi\theta\cdot n^{(p)})$ is precisely the quantity we controlled in \S\ref{sec:scalfib}. Following the arguments, we see that the same proof continues to hold here. In fact, $\tilde{g}_{m,r}(z)$ only has an additional term $\mu_s-\mu_w$, and the proof continues to hold, as no $\gamma_p$ is zero so this term cannot induce cancellations in the polynomial $z^{\gamma_*}\tilde{g}_{m,r}(z)$. Thus, the quantity in \eqref{e:eigenass} is $\le M N^{-1}\to 0$ as required, with the same $M\le 4d \ell_D^{d-1}q$ of the case $\nu=1$. This shows that the assumption of Theorem~\ref{thm:maingen} is satisfied for $\Gamma\mathop\square G_F$.

By \eqref{e:cartop}, the eigenvectors of $H_{\Gamma\mathop\square G_F}(\theta_\fb)$ are just the eigenvectors of $H_{G_F}$ (recall that $H_\Gamma(\theta_\fb)$ is just a scalar $1\times 1$ matrix). They are thus independent of $\theta_\fb$, and so are the eigenprojectors $P_s(\theta_\fb)$. This makes \eqref{e:opnabarpsi} a bit simpler here. If moreover we choose $\psi =\psi_u^{(N)}$ to consist of a tensor basis $\psi_u^{(N)} = \phi_n \otimes w_j$, where $(\phi_n)$ is an orthonormal eigenbasis of $H_{\Gamma}$ on $\Gamma$ and $(w_j)$ is an orthonormal eigenbasis of $H_{G_F}$, then the expression simplifies further. In fact, recalling \eqref{e:u}, we have $(U\psi)_r(v_q) = \frac{1}{N^{d/2}}\sum_k \ee^{-2\pi\ii r\cdot k/N}\phi_n(k_\fa)w_j(v_q) = \widehat{\phi}_n(r)w_j(v_q)$, where $\widehat{\phi}_n(r)$ is the Fourier coefficient of $\phi_n$ in the basis $e_k^{(N)}$ of $\ell^2(\LL_N^d)$. Hence, $(P_s(U\psi)_r)(v_q)  = \widehat{\phi}_n(r)(P_s w_j)(v_q)$. Thus, \eqref{e:opnabarpsi} simplifies to
\[
\sum_{q=1}^\nu \langle a(\cdot+v_q) \rangle\sum_r\sum_{s=1}^{\nu'} |\widehat{\phi}_n(r)|^2|(P_{E_s}w_j)(v_q)|^2 = \sum_{q=1}^\nu \langle a(\cdot+v_q)\rangle \sum_{s=1}^{\nu'}|(P_{E_s}w_j)(v_q)|^2\,,
\]
where we used that $\|\phi_n\|^2=1$. But $w_j$ is an eigenvector, so $P_{E_s}w_j = w_j$ if $E_s=E_j$ and $P_{E_s}w_j = 0$ otherwise. This completes the proof.
\end{proof}

\subsection{Graph decorations}\label{sec:deco}
Another way to create a new graph from given infinite and finite graphs $\Gamma$ and $G_F$ is to simply attach a copy of $G_F$ at each vertex of $\Gamma$. More precisely, we identify a special vertex $o_F\in G_F$ to each $v\in \Gamma$. This process is called \emph{graph decoration}. A very simple example is given in Figure~\ref{fig:deco}. The resulting graph is sometimes denoted by $\Gamma \mathop\triangleleft G_F$ (which reflects the procedure).

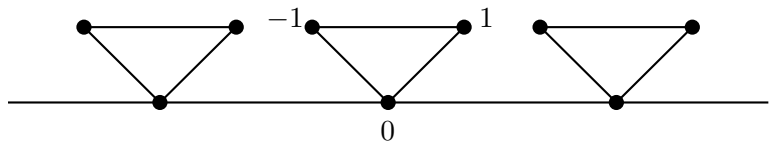
\begin{figure}[h!]
\begin{center}
\setlength{\unitlength}{1cm}
\thicklines
\begin{picture}(1.5,1.5)(-0.3,-0.3)
   \put(-5,0){\line(1,0){10}}
   \put(0,0){\line(1,1){1}}
   \put(0,0){\line(-1,1){1}}
   \put(3,0){\line(1,1){1}}
   \put(3,0){\line(-1,1){1}}
   \put(-3,0){\line(1,1){1}}
   \put(-3,0){\line(-1,1){1}}
	 \put(-1,1){\line(1,0){2}}
	 \put(2,1){\line(1,0){2}}
	 \put(-4,1){\line(1,0){2}}
	 \put(0,0){\circle*{.2}}
	 \put(3,0){\circle*{.2}}
	 \put(-3,0){\circle*{.2}}
	 \put(4,1){\circle*{.2}}
	 \put(2,1){\circle*{.2}}
	 \put(1,1){\circle*{.2}}
	 \put(-1,1){\circle*{.2}}
	 \put(-2,1){\circle*{.2}}
	 \put(-4,1){\circle*{.2}}
	 \put(-1.6,1){$-1$}
	 \put(1.2,1){$1$}
	 \put(-0.1,-0.5){$0$}
\end{picture}
\caption{Decorating $\Z$ with triangles. The values of an eigenfunction are shown (it is then extended by zero).}\label{fig:deco}
\end{center}
\end{figure} 

In contrast to Cartesian products, this process can be problematic for delocalization. For example, as shown in Figure~\ref{fig:deco}, this can create compactly supported eigenfunctions. The corresponding eigenvalue is a flat band, i.e. an infinitely degenerate eigenvalue. The example in Fig.~\ref{fig:deco} has the Floquet eigenvalues $\{-1,\frac{2\cos 2\pi\theta + 1 \pm \sqrt{4\cos^2 2\pi\theta - 4\cos 2\pi\theta + 9}}{2}\}$. This generates the spectrum of $H=\cA$ consisting of two bands which do not intersect. This spectrum is not very nice as the eigenvalue $-1$ is embedded in the left band, as can be seen by taking $\theta = \frac{1}{4}$.

It may be interesting to observe that in general, if $\Gamma$ is a periodic graph having $\nu=1$, then $\Gamma$, $\Gamma\mathop\square G_F$ and $\Gamma \mathop\triangleleft G_F$ are all ``loop graphs'' in the sense of Korotyaev and  Saburova \cite{KorSa}. This class of graphs was singled out in \cite{KorSa} for being more amenable to spectral analysis. We see that not all graphs in this class are quantum ergodic.

\begin{proof}[Proof of Proposition~\ref{prp:zno}]
For the graph in Fig.~\ref{fig:deco}, we have $|\Gamma_N|=3N$, and on $\Gamma_N$, we may construct $N$ localized eigenfunctions $f_j$, one on each triangle, each supported on only two vertices. Let $N$ be even and take the locally constant observable $a$ which is identically $1$ on triangles attached to even vertices, and identically zero on triangles attached to odd vertices. Then $\langle a\rangle =\frac{1}{2}$. On the other hand, if we normalize the eigenfunctions $f_j$ so that their values are $(\frac{1}{\sqrt{2}},\frac{-1}{\sqrt{2}},0,0,\dots,0)$, then $\langle f_{2j},a f_{2j}\rangle = \sum_v a(v)|f_{2j}(v)|^2 = 1$, while $\langle f_{2j+1},af_{2j+1}\rangle =0$ for each $j$. Hence,
\begin{align*}
\frac{1}{|\Gamma_N|}\sum_{u\in \Gamma_N} |\langle \psi_u^{(N)},a\psi_u^{(N)}\rangle -\langle a\rangle|^2 &\ge \frac{1}{3N}\sum_{j=1}^N |\langle f_j,af_j\rangle - \langle a\rangle|^2 \\
&= \frac{1}{3N}\Big[\frac{N}{2}\Big(\Big|1-\frac{1}{2}\Big|^2 + \Big|0-\frac{1}{2}\Big|^2\Big)\Big] = \frac{1}{12}\,.
\end{align*}

\end{proof}

\subsection{More product operations}\label{sec:tenstro}
Further operations to construct new graphs from old are the \emph{tensor product} and the \emph{strong product} of graphs.

\subsubsection{Strong products}\label{sec:stro} The strong product $G\boxtimes H$ has vertex set $V(G)\times V(H)$, with $(u,v)\sim (u',v')$ iff ($u=u'$ and $v\sim v'$) or ($u\sim u'$ and $v=v'$) or ($u\sim u'$ and $v\sim v'$). We thus add more edges to the Cartesian product. 

This operation is not as well behaved as the Cartesian one. For example, consider $\Z\boxtimes P_2$, where $P_2$ is a $2$-path. The result (Figure~\ref{fig:box})
\begin{figure}[h!]
\begin{center}
\setlength{\unitlength}{1cm}
\thicklines
\begin{picture}(1.3,1.3)(-1.1,-1.1)
   \put(-5,0){\line(1,0){10}}
	 \put(-5,-1){\line(1,0){10}}
	 \put(-5,-1){\line(2,1){2}}
	 \put(-5,0){\line(2,-1){2}}
	 \put(-3,-1){\line(0,1){1}}
	 \put(-3,-1){\line(2,1){2}}
	 \put(-3,0){\line(2,-1){2}}
	 \put(-1,-1){\line(0,1){1}}
	 \put(-1,-1){\line(2,1){2}}
	 \put(-1,0){\line(2,-1){2}}
	 \put(1,-1){\line(0,1){1}}
		\put(1,-1){\line(2,1){2}}
	 \put(1,0){\line(2,-1){2}}
	 \put(3,-1){\line(0,1){1}}
		\put(3,-1){\line(2,1){2}}
	 \put(3,0){\line(2,-1){2}}
	 \put(-1,-1){\circle*{.2}}
	 \put(-1,0){\circle*{.2}}
	 \put(-3,-1){\circle*{.2}}
	 \put(-3,0){\circle*{.2}}
	 \put(1,-1){\circle*{.2}}
	 \put(1,0){\circle*{.2}}
	 \put(3,-1){\circle*{.2}}
	 \put(3,0){\circle*{.2}}
	 \put(0.9,0.2){$1$}
	 \put(0.8,-1.5){$-1$}
\end{picture}
\caption{$\Z\boxtimes P_2$. An eigenfunction localized on two vertices is shown.}\label{fig:box}
\end{center}
\end{figure}
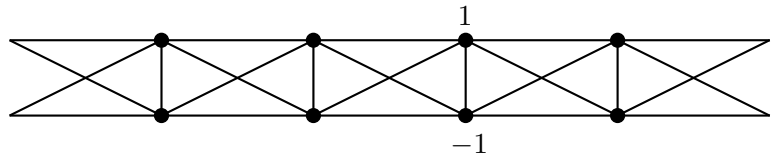
is an infinite sequence of boxes~$\boxtimes$. Unlike the ladder, this graph has some point spectrum. In fact, the Floquet matrix here is $H(\theta_\fb) = \begin{pmatrix} 2\cos 2\pi\theta & 1+2\cos 2\pi\theta\\ 2\cos 2\pi\theta +1 & 2\cos 2\pi\theta \end{pmatrix}$, with eigenvalues $\{-1,1+4\cos 2\pi\theta\}$. Quantum ergodicity is violated (use the eigenfunction shown in Figure~\ref{fig:box} and argue as in \S\ref{sec:deco}).

See \S\ref{sec:boxes} for a further analysis when we add a potential.

Still, this product sometimes behaves well. For example, $\Z\boxtimes \Z$ gives the king's graph, which is quantum ergodic since it is periodic with $\nu=1$.

\subsubsection{Tensor products}
Next, the tensor product $G\times H$ has vertex set $V(G)\times V(H)$, with $(u,v)\sim (u',v')$ iff ($u\sim u'$ and $v\sim v'$). Equivalently, $\cA_{G\times H} = \cA_G\otimes \cA_H$. The edges of this product are precisely the ones we added to the Cartesian product when discussing strong products.

The product of two connected graphs is not necessarily connected. For example, the tensor product of two path graphs of length $2$ $\{a,b\}$ and $\{v,w\}$ gives the union of the two paths $\{(a,v),(b,w)\}$ and $\{(a,w),(b,v)\}$. To consider a product graph of the form $\Gamma \times G_F$ for quantum ergodicity, where $\Gamma$ is a quantum ergodic graph and $G_F$ is some finite graph, we first need $\Gamma \times G_F$ to be connected. It turns out this is satisfied if and only if either $\Gamma$ or $G_F$ contains an odd cycle, see \cite{Weic}.

Assume now that we are given a periodic $\Gamma$ with $\nu=1$, for simplicity. Just like Cartesian products, the tensor structure of the adjacency matrix translates well into the Floquet fibers. To see this, it is best to first picture the product operation. Geometrically, we simply consider the $G_F$-layers structure of Cartesian products, but then we remove all edges and add instead the following ones~: a given $(u,v)$ in a $G_F$ layer is connected to all vertices $(u',v')$, where $u'$ is in a different $G_F$ layer and $v'\sim v$ in $G_F$. Note that $V_f=G_F$ contains no edges. Instead, if we ``project'' the edges going from a neighboring $G_F$ layer to the starting one, we obtain the finite graph $G_F$ that we started with. We may also endow $G_F$ with some potential $Q$ which is copied across the layers.

By definition \eqref{e:htheta}, we have $H(\theta_\fb)f(u,v) = \sum_{u'\sim u,v'\sim v} \ee^{\ii\theta_\fb\cdot \lfloor u'\rfloor_\fa} f(u,v') + Q_v f(u,v) = H_\Gamma(\theta_\fb) \otimes H_{G_F} f(u,v)$, where we used here that $\{(u',v')\}_\fa = (u,v')$ and $\lfloor (u',v')\rfloor_\fa = \lfloor u'\rfloor_\fa$ by construction. This shows that $H_{\Gamma\times G_F}(\theta_\fb) = H_\Gamma(\theta_\fb) \otimes H_{G_F}$. Consequently,
\begin{equation}\label{e:tensorspec}
\sigma(H_{\Gamma \times G_F}(\theta_\fb)) = \{\mu_j E_\Gamma(\theta_\fb)\}_{j=1}^\nu,
\end{equation}
where $\mu_j$ are the eigenvalues of $H_{G_F}$. Note that if $\mu_j=0$ for some $j$, then this creates a flat band $\{0\}$ for $H_{\Gamma\times G_F}$, i.e. an infinitely degenerate eigenvalue.

We now consider the special case of $\Z\times G_F$. So $E_\Gamma(\theta_\fb) = 2\cos 2\pi\theta$.

\begin{proof}[Proof of Proposition~\ref{prp:counterass}] 
To construct a counterexample, we take $G_F$ such that
\begin{enumerate}[(i)]
\item $G_F$ is not bipartite,
\item $0\notin \sigma(\cA_{G_F})$,
\item there exists $s$ such that $\mu_s$ and $-\mu_s$ belong to $\sigma(\cA_{G_F})$. 
\end{enumerate}

Point (i) is necessary to make $\Z\times G_F$ connected, (ii) is necessary to avoid a point spectrum $\{0\}$, and (iii) is what will contradict \eqref{e:eigenass}.

We take $G_F$ as the butterfly graph

\begin{figure}[h!]
\begin{center}
\setlength{\unitlength}{0.6cm}
\thicklines
\begin{picture}(6,6)(-3.8,-3.8)
	 \put(2,2){\circle*{.2}}
	 \put(3,3){\circle*{.2}}
	 \put(6,2){\circle*{.2}}
	 \put(5,3){\circle*{.2}}
	 \put(4,2.5){\circle*{.2}}
	 \put(2,-1){\textcolor{red}{\circle*{.2}}}
	 \put(3,0){\textcolor{red}{\circle*{.2}}}
	 \put(6,-1){\textcolor{red}{\circle*{.2}}}
	 \put(5,0){\textcolor{red}{\circle*{.2}}}
	 \put(4,-0.5){\textcolor{red}{\circle*{.2}}}
	 \put(2,-4){\circle*{.2}}
	 \put(3,-3){\circle*{.2}}
	 \put(6,-4){\circle*{.2}}
	 \put(5,-3){\circle*{.2}}
	 \put(4,-3.5){\circle*{.2}}
	 \put(2,-4){\line(1,4){1}}
	 \put(1.95,-4){\line(3,5){2.1}}
	 \put(6,-4){\line(-1,4){1}}
	 \put(6.05,-4){\line(-3,5){2.1}}
	 \put(3,-3){\line(-1,2){1}}
	 \put(3,-3){\line(2,5){1}}
	 \put(5,-3){\line(1,2){1}}
	 \put(5,-3){\line(-2,5){1}}
	 \put(3.95,-3.5){\line(-1,4){0.9}}
	 \put(3.95,-3.5){\line(-3,4){1.9}}
	 \put(4.05,-3.5){\line(3,4){1.9}}
	 \put(4.05,-3.5){\line(1,4){0.9}}
	 	 \put(2,-1){\line(1,4){1}}
	 \put(1.95,-1){\line(3,5){2.1}}
	 \put(6,-1){\line(-1,4){1}}
	 \put(6.05,-1){\line(-3,5){2.1}}
	 \put(3,0){\line(-1,2){1}}
	 \put(3,0){\line(2,5){1}}
	 \put(5,0){\line(1,2){1}}
	 \put(5,0){\line(-2,5){1}}
	 \put(3.95,-0.5){\line(-1,4){0.9}}
	 \put(3.95,-0.5){\line(-3,4){1.9}}
	 \put(4.05,-0.5){\line(3,4){1.9}}
	 \put(4.05,-0.5){\line(1,4){0.9}}
		 \put(-6,-1){\line(0,1){1}}
	 \put(-6,-1){\line(1,0.5){2}}
	 \put(-6,0){\line(1,-0.5){2}}
	 \put(-4,-1){\line(0,1){1}}
	 \put(-6,-1){\circle*{.2}}
	 \put(-6,0){\circle*{.2}}
	 \put(-4,-1){\circle*{.2}}
	 \put(-4,0){\circle*{.2}}
	 \put(-5,-0.5){\circle*{.2}}
	 \put(-6.8,0){\small{$v_1$}}
	 \put(-6.8,-1){\small{$v_2$}}
	 \put(-5.2,-0.2){\small{$v_3$}}
	 \put(-3.8,0){\small{$v_4$}}
	 \put(-3.8,-1){\small{$v_5$}}
\end{picture}
\caption{The butterfly graph $G_F$ (left) and part of the tensor product $\Z\times G_F$ (right). A fundamental set is colored in red}
\end{center}
\end{figure}
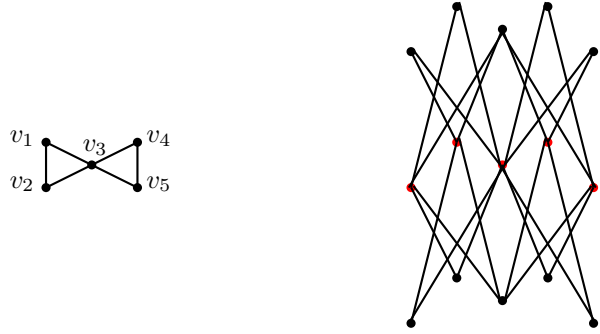

Since $\cA_{G_F}$ is a $5\times 5$ matrix, we can compute its eigenvalues and eigenvectors explicitly and find the following:
\[
\mu_1 = \frac{1+\sqrt{17}}{2}\,, \qquad \mu_2 = \frac{1-\sqrt{17}}{2}\,, \qquad \mu_3=-1\,, \qquad \mu_4=-1\,, \qquad \mu_5=1
\]
\[
w_1 = c_1\Big(1,1,\frac{-1+\sqrt{17}}{2},1,1\Big)\,, \qquad w_2 = c_2\Big(1,1,\frac{-1-\sqrt{17}}{2},1,1\Big) \,,
\]
\[
w_3=\frac{1}{\sqrt{2}}(0,0,0,-1,1)\,, \qquad w_4=\frac{1}{\sqrt{2}}(-1,1,0,0,0)\,, \qquad w_5 = \frac{1}{2}(-1,-1,0,1,1)
\]
for normalization constants $c_1,c_2$. We actually only need $w_4,w_5$ for the following argument, it is immediate to check that they are eigenvectors for $\mu_4,\mu_5$, respectively.

We see properties (i)--(iii) are satisfied, take e.g. $\mu_s = 1$.

By \eqref{e:tensorspec}, $\sigma(\cA_{\Z\times G_F}(\theta_\fb))$ is just $\{2\mu_j \cos 2\pi\theta\}$, where $\mu_j$ runs over the above list of eigenvalues. It follows that $\sigma(\cA_{\Z\times G_F})$ is purely absolutely continuous (as each Floquet eigenvalue is analytic and nonconstant, see \cite[Th. XIII.86]{RS4}). The graph $\Z\times G_F$ is also connected, since $[\![-n,n]\!]\times G_F$ is connected for any $n$ by \cite{Weic}.

If $\mu_s = 1$ and $\mu_w = -1$, we find that
\[
E_s(\theta_\fb+\alpha_\fb) - E_w(\theta_\fb) = \mu_s(2\cos(2\pi(\theta+\alpha)) + 2\cos 2\pi\theta) = 4\mu_s \cos\pi(2\theta+\alpha)\cos\pi\alpha \,.
\]
This is zero if $\alpha=\frac{1}{2}$, for all $\theta$. This suffices to contradict \eqref{e:eigenass}. In fact, taking $m=\frac{N}{2}\in \LL_N$ assuming $N$ is even, the fraction in \eqref{e:eigenass} is equal to $1$ and does not vanish. 

We now show the tensor product $\Z\times G_F$ is not quantum ergodic. The hint for the choice of the observable comes from Remark~\ref{rem:nec}. Namely, consider $a(k+v_q) = \ee^{2\pi \ii mk/N}\delta_{v_1}(v_q)$. Then $\langle a(\cdot+v_q)\rangle = 0$ for all $v_q$, so $\langle \psi,\opn(\overline{a})\psi\rangle=0$ by \eqref{e:opnabarpsi}. We choose the problematic value of $m$, namely $m=\frac{N}{2}$, so we take $a(k+v_q) := \ee^{\pi\ii k}\delta_{v_1}(v_q)$.

Now, choose $\phi_n(k) = \frac{1}{\sqrt{N}}\ee^{2\pi\ii nk/N}$ as an eigenbasis for $\cA_{P_N}$ with periodic conditions and consider the orthonormal sequence
\[
g_n= \frac{\phi_{n}\otimes w_4+\phi_{n+\frac{N}{2}}\otimes w_5}{\sqrt{2}}
\]
in $\Gamma_N=P_N\otimes G_F$, for $n=0,\dots, \frac{N}{2}-1$, with eigenvalue $-\lambda_n=-2\cos\frac{2\pi n}{N}$.

Since $\langle \psi,\opn(\overline{a})\psi\rangle=0$, it suffices to show that $\frac{1}{|\Gamma_N|}\sum_{u\in \Gamma_N} |\langle \psi_u,a\psi_u\rangle|^2$ does not converge to zero. We have
\begin{multline*}
\langle g_n,ag_n\rangle = \sum_{k=0}^{N-1}\sum_{q=1}^5 a(k+v_q)|g_n(k+v_q)|^2 = \frac{1}{2}\sum_{k=0}^{N-1}\ee^{\pi\ii k}|\phi_{n}(k)w_4(v_1)+\phi_{n+\frac{N}{2}}(k)w_5(v_1)|^2\\
= \frac{1}{2N}\sum_{k=0}^{N-1}\ee^{\pi\ii k}\Big|\frac{\ee^{2\pi\ii nk/N}}{\sqrt{2}}+ \frac{\ee^{2\pi\ii(n+\frac{N}{2})k/N}}{2}\Big|^2 = \frac{1}{4N}\sum_{k=0}^{N-1}\ee^{\pi\ii k} \Big|1+\frac{\ee^{\pi\ii k}}{\sqrt{2}}\Big|^2\\
=\frac{1}{4N}\sum_{k=0}^{N-1}\ee^{\pi\ii k}\Big(\frac{3}{2}+\frac{\ee^{\pi\ii k}+\ee^{-\pi\ii k}}{\sqrt{2}}\Big) =  \frac{1}{4\sqrt{2}}.
\end{multline*}
Thus, by completing the orthonormal family $(g_n)$ to an o.n.b. $(\psi_u)$, we get
\[
\frac{1}{|\Gamma_N|}\sum_{u\in\Gamma_N} |\langle \psi_u,a\psi_u\rangle|^2 \ge \frac{1}{5N}\sum_{n=0}^{ \frac{N}{2} -1} |\langle g_n,ag_n\rangle|^2 = \frac{N/2}{5N}\cdot\frac{1}{32}= \frac{1}{320}.
\]
This completes the proof.
\end{proof}

Note that this counterexample violates \eqref{e:eigenass} for some $w\neq s$. For $w=s$, \eqref{e:tensorspec} tells us that in general $E_s(\theta_\fb+\alpha_\fb)-E_s(\theta_\fb) = \mu_s(E_\Gamma(\theta_\fb+\alpha_\fb)-E_\Gamma(\theta_\fb))$, which is controlled in \S\ref{sec:scalfib}. Hence, Theorem~\ref{thm:maingen}(ii) implies that if $\Gamma$ is any periodic graph with $\nu=1$, and if $G_F$ is a finite nonbipartite graph such that $0\notin\sigma(\cA_{G_F})$, then quantum ergodicity is satisfied for the tensor product $\Gamma\times G_F$ \emph{for the class of locally constant observables.}

\section{Concrete examples} \label{sec:examp}

\subsection{Graphs with scalar fibers}\label{sec:scalexa}
For the adjacency matrix $H = \cA$ on $\Z^d$ or the triangular lattice (sometimes called hexagonal, see \cite[Fig. 3]{KorSa}) where each vertex has 6 neighbors, or the king's graph (sometimes called EHM lattice), we have $\nu=1$ so Theorem~\ref{thm:mainnu1} applies.

The family of periodic graphs having $\nu=1$ is quite rich. For example, one can consider $\Z$ and add edges up to some fixed distance $k$ from each vertex. More precisely,
\[
(\cA f)(n) = f(n-k)+f(n-k+1)+\dots+f(n+k-1)+f(n+k) \,.
\]
Then $V_f = \{0\}$, $\fa_1 = \mathfrak{e}_1$ and $H(\theta_\fb) = 2\cos 2\pi\theta + 2\cos 4\pi\theta+\dots+ 2\cos 2\pi k\theta$. Similar variations can be performed on $\Z^d$.

We remark that the connectedness of $\Gamma$ is important. For example, if we consider $\Z$ with $(\cA f)(n) = f(n-2)+f(n+2)$, then the graph is disconnected (there are two copies of $\Z$, for the even and odd vertices, respectively). Here, $V_f = \{0\}$, $\fa_1=\mathfrak{e}_1$ and $H(\theta_\fb) = 2\cos 4\pi\theta$, which does not obey \eqref{e:eigenass}, since for $\alpha=\frac{1}{2}$, we have $E(\theta_\fb + \alpha_\fb) = E(\theta_b)$ for all $\theta$.

\subsection{Honeycomb lattice}\label{sec:honey} 
Consider the honeycomb lattice (\cite[Fig. 7]{KorSa}, a.k.a graphene or hexagonal lattice) where each vertex has 3 neighbors. Here $\nu=2$, $H(\theta_{\fb}) = \begin{pmatrix} 0& \xi(\theta_{\fb})\\ \overline{\xi(\theta_{\fb})} &0 \end{pmatrix}$, where $\xi(\theta_{\fb}) = 1+\ee^{-\ii\theta_{\fb}\cdot\fa_1}+\ee^{-\ii\theta_{\fb}\cdot\fa_2}$ for the crystal basis $\fa_1=a(1,0)$, $\fa_2=\frac{a}{2}(1,\sqrt{3})$, $a>0$. This gives the eigenvalues $\pm |\xi(\theta_{\fb})| = \pm\sqrt{3+2\cos 2\pi\theta_1+2\cos 2\pi\theta_2 + 2\cos 2\pi(\theta_1-\theta_2)}$. Assumption \eqref{e:eigenass} is clearly satisfied if $w\neq s$ as the bands only meet at $0$ (for $(\theta_1,\theta_2)=(\frac{2}{3},\frac{1}{3})$). On the other hand, we can control the event that $|\xi(\theta_\fb+\alpha_\fb)| = |\xi(\theta_\fb)|$ by squaring, deducing as a special consequence of the arguments in \S\ref{sec:scalfib} that \eqref{e:eigenass} is satisfied. This shows that Theorem~\ref{thm:maingen} holds true. Let us investigate \eqref{e:opnabarpsi}.

The eigenvectors $w_{\pm}(\theta_{\fb}) = \frac{1}{\sqrt{2}}(1,\pm \ee^{-\ii\phi(\theta_{\fb})})^\intercal$, where $\phi(\theta_{\fb})$ is the argument of $\xi(\theta_{\fb})$. So $P_{\pm}(\theta_{\fb})f(v_1) = \frac{f(v_1)\pm \ee^{\ii\phi(\theta_{\fb})} f(v_2)}{2}$, $P_{\pm}(\theta_{\fb})f(v_2) = \frac{f(v_1)\pm \ee^{\ii\phi(\theta_{\fb})} f(v_2)}{2}(\pm \ee^{-\ii\phi(\theta)_{\fb}})$. It follows that $|P_+f(v_1)|^2 + |P_-f(v_1)|^2 = \frac{|f(v_1) +\ee^{\ii\phi(\theta_{\fb})}f(v_2)|^2+|f(v_1)-\ee^{\ii\phi(\theta_{\fb})}f(v_2)|^2}{4} = \frac{|f(v_1)|^2+|f(v_2)|^2}{2} = \frac{\|f\|^2}{2} = |P_+f(v_2)|^2+|P_-f(v_2)|^2$. 

We showed that for the honeycomb lattice, \eqref{e:opnabarpsi} reduces to
\[
\sum_{q=1}^2\sum_{r\in\LL_N^d} \frac{\|(U\psi)_r\|_{\C^\nu}^2}{2}\langle a(\cdot+v_q)\rangle = \frac{\langle a(\cdot+v_1)\rangle + \langle a(\cdot+v_2)\rangle}{2}\|\psi\|^2
\]
which is the uniform average.

\subsection{Ladder graph} \label{sec:lad}
Consider the ladder graph $\Z\mathop\square P_2$ in Figure~\ref{fig:lad}. As a Cartesian product, we already know that Proposition~\ref{prp:cart} holds true, but we show here that we always get the uniform average in this example.

We have $H(\theta_\fb)f(v_1) = \ee^{2\pi\ii\theta}f(v_1) + \ee^{-2\pi\ii \theta}f(v_1) + f(v_2)$ and $H(\theta_\fb)f(v_2) = \ee^{2\pi\ii\theta}f(v_2) + \ee^{-2\pi\ii\theta}f(v_2) + f(v_1)$. Thus, $H(\theta_\fb) = \begin{pmatrix} 2\cos 2\pi\theta & 1\\ 1& 2\cos 2\pi\theta \end{pmatrix}$. The eigenvalues are $E_{\pm}(\theta_\fb) = 2\cos 2\pi\theta \pm 1$. Clearly $(1,1)$ and $(1,-1)$ are eigenvectors. So the eigenprojectors are $P_{\pm}(\theta_\fb) f = \langle w_{\pm}, f\rangle w_{\pm}$, with $w_{\pm} = \frac{1}{\sqrt{2}}(1,\pm 1)$, independently of $\theta$. Thus, $P_{\pm}f(v_1) = \frac{f(v_1) \pm f(v_2)}{2}$ and $P_{\pm}f(v_2) = - \frac{f(v_1)\pm f(v_2)}{2}$. As in the honeycomb lattice, we deduce that \eqref{e:opnabarpsi} reduces to $\frac{\langle a(\cdot+v_1)\rangle + \langle a(\cdot+v_2)\rangle}{2}\|\psi\|^2$.

If we endow $P_2$ with a potential $Q_\bullet,Q_\circ$, then we get a ladder with a potential coming in two parallel sheets, the upper sheet only containing $Q_\bullet$, the lower only $Q_\circ$. The construction can be generalized to $\Z\mathop\square P_k$ to create an infinite $k$-strip. Proposition~\ref{prp:cart} continues to apply, but the average may be non-uniform.

\subsection{Periodic potentials on the integer lattice}\label{sec:1d}
Consider $\Z$ endowed with a periodic potential taking $\nu$ values $Q_i$. We have $V_f = \{1,\dots,\nu\}$, $\fa_1 = \nu \mathfrak{e}_1$ and $\fb_1 = \frac{2\pi}{\nu}\mathfrak{e}_1$. 

Now $H(\theta_\fb)f(1) = Q_1 f_1 + f_2+ \ee^{-2\pi\ii\theta}f(\nu)$, $H(\theta_\fb)f(i) = Q_i f_i + f_{i-1} + f_{i+1}$ for $1<i<\nu$ and $H(\theta_\fb)f(\nu) = Q_\nu f_\nu + f_{\nu-1} + \ee^{2\pi\ii\theta} f_1$. We thus have
\[
H(\theta_\fb) = \begin{pmatrix} Q_1 & 1 & 0 & \cdots & \ee^{-2\pi\ii\theta}\\ 1 & Q_2 & 1 & & 0\\ & & \ddots\\ & & & &1 \\ \ee^{2\pi\ii\theta} & 0& & 1& Q_\nu \end{pmatrix} \,.
\]

Let $z = \ee^{2\pi\ii\theta}$. Expanding the determinant of the characteristic polynomial $p(z;\lambda)$ in detail, we see that \cite[Lemma 3.1]{GKT}
\begin{equation}\label{e:charpol1d}
p(z;\lambda) = \Delta(\lambda)-z-z^{-1}
\end{equation}
for some polynomial $\Delta(\lambda,Q)$. This splitting into pure $\lambda$ and $z$ parts is specific to one dimension. 

Now fix $\alpha\neq 0$, let $\zeta = \ee^{2\pi\ii\alpha}$ and suppose that $E_s(\theta_\fb+\alpha_\fb) = E_w(\theta_\fb)$ for some $s,w$. Then $\lambda = E_s(\theta_\fb+\alpha_\fb)$ solves \eqref{e:charpol1d}. On the other hand, $\lambda=E_s(\theta_\fb+\alpha_\fb)$ is also a root of the characteristic polynomial of $H(\theta_\fb+\alpha_\fb)$, which is
\[
p(z\zeta;\lambda) = \Delta(\lambda) -z\zeta - (z\zeta)^{-1}\,.
\]

For this choice of $\lambda$ we thus have $p(\lambda;z) = p(\lambda;z\zeta)=0$. So $z+z^{-1}=z\zeta+(z\zeta)^{-1}$. This yields a quadratic expression for $z$. Hence, for any fixed $\alpha\neq 0$, there are at most two $\theta$ such that $E_s(\theta_\fb+\alpha_\fb) = E_w(\theta_\fb)$. This implies \eqref{e:eigenass}.

\smallskip

The case of $\cA+Q$ on $\Z^d$, $d>1$, with $Q(n+p_j\mathfrak{e}_j)=Q(n)$, is more delicate. The criterion has been established in \cite{Wen2} using the point of view of Bloch varieties; see \S~\ref{sec:irred} for some background. Here we simply mention that for this model, it is equivalent to study
\[
\widetilde{H}(\theta) = D_\theta +B_Q\,,
\]
on $\ell^2(V_f)$, where $D_\theta$ is a diagonal operator and $B_Q$ is a convolution given by
\[
(D_\theta f)(u) = \bigg(\sum_{j=1}^d 2\cos2\pi \Big(\frac{u_j+\theta_j}{p_j}\Big)\bigg)f(u)\,, \qquad (B_Qf)(u) = \sum_{v_q\in V_f}\widehat{Q}\Big(\frac{u-v_q}{p}\Big)f(v_q)\,,
\]
with $\widehat{Q}(\sigma)=\frac{1}{\nu}\sum_{v_n\in V_f} Q(v_n)\ee^{-2\pi\ii\sigma\cdot v_n}$ and $\frac{u}{p} := (\frac{u_1}{p_1},\dots,\frac{u_d}{p_d})$. 

Note that $V_f = \times_{j=1}^d[\![0,p_j-1]\!]$, so that $\nu = \prod_{j=1}^d p_j$. It is not difficult to show that our operator $H(\theta_\fb)$ is unitarily equivalent to $\widetilde{H}(\theta)$, with
\[
H(\theta_\fb) = \mathcal{F}_\theta^{-1}\widetilde{H}(\theta)\mathcal{F}_\theta\,,
\]
where $\mathcal{F}_\theta:\ell^2(V_f)\to\ell^2(V_f)$ is given by
\[
(\mathcal{F}_\theta f)(u) = \frac{1}{\sqrt{\nu}}\sum_{v_q\in V_f}\ee^{-2\pi\ii(\frac{u+\theta}{p})\cdot v_q}f(v_q)\,.
\]
This equivalence is used in the proof of \cite{Wen2}.

\smallskip

Back to $d=1$, let us examine \eqref{e:opnabarpsi} for $\Z$ with a $2$-periodic potential $Q_\bullet,Q_\circ$. Here $H(\theta_{\fb}) = \begin{pmatrix} Q_\bullet & 1+\ee^{-2\pi\ii\theta}\\ 1+\ee^{2\pi\ii\theta}& Q_\circ\end{pmatrix}$. The eigenvalues solve $(Q_\bullet-\lambda)(Q_\circ-\lambda)-(2+2\cos 2\pi\theta)=0$, so $E_\pm(\theta_{\fb}) = \frac{Q_\bullet+Q_\circ \pm c}{2}$, with $w_\pm = (\frac{Q_\bullet-Q_\circ\pm c}{2(1+\ee^{2\pi\ii\theta})},1)$ and $c = \sqrt{(Q_\bullet-Q_\circ)^2+16\cos^2  \pi\theta}$. 

After some tedious calculations, we conclude that \eqref{e:opnabarpsi} takes the form
\begin{multline}\label{e:av2per}
\langle \psi,\opn(\overline{a})\psi\rangle = \sum_{q=1}^2\langle a(\cdot+v_q)\rangle \sum_{r=0}^{N-1}\Big[|P_+\Big(\frac{r_{\fb}}{N}\Big)(U\psi)_r(v_q)|^2+|P_-\Big(\frac{r_{\fb}}{N}\Big)(U\psi)_r(v_q)|^2\Big]\\
= \langle a(\cdot)\rangle \sum_{r=0}^{N-1} \Big[\frac{8\cos^2 \frac{\pi r}{N}+(Q_\circ-Q_\bullet)^2}{16\cos^2 \frac{\pi r}{N}+(Q_\circ-Q_\bullet)^2} |(U\psi)_r(0)|^2 + \frac{8\cos^2\frac{\pi r}{N}}{16\cos^2 \frac{\pi r}{N}+(Q_{\circ}-Q_\bullet)^2} |(U\psi)_r(1)|^2\\
 - \frac{2(Q_\circ-Q_\bullet)}{16\cos^2 \frac{\pi r}{N}+(Q_\circ-Q_\bullet)^2} \Re (1+\ee^{-\frac{2\pi\ii r}{N}})(U\psi)_r(0)\overline{(U\psi)_r(1)}\Big]\\
+\langle a(\cdot+1)\rangle\sum_{r=0}^{N-1}\Big[\frac{8\cos^2 \frac{\pi r}{N}}{16\cos^2 \frac{\pi r}{N}+(Q_{\circ}-Q_\bullet)^2} |(U\psi)_r(0)|^2 + \frac{8\cos^2 \frac{\pi r}{N} + (Q_\circ-Q_\bullet)^2}{16\cos^2 \frac{\pi r}{N}+(Q_\circ-Q_\bullet)^2}|(U\psi)_r(1)|^2\\
 + \frac{2(Q_\circ-Q_\bullet)}{16\cos^2 \frac{\pi r}{N}+(Q_\circ-Q_\bullet)^2} \Re (1+\ee^{-\frac{2\pi\ii r}{N}})(U\psi)_r(0)\overline{(U\psi)_r(1)}\Big]\,.
\end{multline}

Note that if $\langle a(\cdot)\rangle= \langle a(\cdot+1)\rangle$, this indeed reduces to $\langle a(\cdot)\rangle \|\psi\|^2$. 

Let us study the expression in the limit $|Q_\circ-Q_\bullet|\to\infty$. We obtain
\[
\lim_{|Q_\circ-Q_\bullet|\to\infty}\langle \psi,\opn(\overline{a})\psi\rangle = \langle a(\cdot)\rangle\sum_{r=0}^{N-1} |(U\psi)_r(0)|^2 + \langle a(\cdot+1)\rangle \sum_{r=0}^{N-1}|(U\psi)_r(1)|^2\,.
\]
Here, $(U\psi)_r(0) = \frac{1}{\sqrt{N}} \sum_{k=0}^{N-1}\ee^{\frac{-2\pi\ii rk}{N}} \psi(2k)$ and $(U\psi)_r(1) = \frac{1}{\sqrt{N}}\sum_{k=0}^{N-1}\ee^{\frac{-2\pi\ii rk}{N}} \psi(2k+1)$. It follows that
\[
\lim_{|Q_\circ-Q_\bullet|\to\infty} \langle \psi,\opn(\overline{a})\psi\rangle = \langle a(\cdot)\rangle \sum_{k=0}^{N-1} |\psi(2k)|^2 + \langle a(\cdot+1)\rangle \sum_{k=0}^{N-1}|\psi(2k+1)|^2\,.
\]

\subsection{Cylinders}\label{sec:cyl}
Consider the Cartesian product $\Gamma = \Z \mathop\square C_4$, where $C_4$ is the $4$-cycle. 

\begin{figure}[h!]
\begin{center}
\setlength{\unitlength}{1cm}
\thicklines
\begin{picture}(1.8,1.8)(-0.8,-0.8)
   \put(-5,0){\line(1,0){10}}
	 \put(-5,-1){\line(1,0){10}}
	 \put(-3,-1){\line(0,1){1}}
	 \put(-1,-1){\line(0,1){1}}
	 \put(1,-1){\line(0,1){1}}
	 \put(3,-1){\line(0,1){1}}
	 \put(-1,-1){\circle*{.2}}
	 \put(-1,0){\circle*{.2}}
	 \put(-3,-1){\circle*{.2}}
	 \put(-3,0){\circle*{.2}}
	 \put(1,-1){\circle*{.2}}
	 \put(1,0){\circle*{.2}}
	 \put(3,-1){\circle*{.2}}
	 \put(3,0){\circle*{.2}}
	 \put(-3,0){\line(1,1){1}}
	 \put(-1,0){\line(1,1){1}}
	 \put(1,0){\line(1,1){1}}
	 \put(3,0){\line(1,1){1}}
	 \put(-4,1){\line(1,0){10}}
	 \put(-2,1){\circle*{.2}}
	 \put(0,1){\circle*{.2}}
	 \put(2,1){\circle*{.2}}
	 \put(4,1){\circle*{.2}}
	 \multiput(4,1)(0,-0.2){5}{\line(0,-1){0.1}}
 	 \multiput(-2,1)(0,-0.2){5}{\line(0,-1){0.1}}
  \multiput(0,1)(0,-0.2){5}{\line(0,-1){0.1}}
	 \multiput(2,1)(0,-0.2){5}{\line(0,-1){0.1}}
	 \multiput(4,1)(0,-0.2){5}{\line(0,-1){0.1}}
	 \multiput(-3,-1)(0.2,0.2){5}{\line(1,1){0.1}}
	\multiput(-1,-1)(0.2,0.2){5}{\line(1,1){0.1}}
	\multiput(1,-1)(0.2,0.2){5}{\line(1,1){0.1}}
	\multiput(3,-1)(0.2,0.2){5}{\line(1,1){0.1}}
	\put(-2,0){\circle{.2}}
	\put(0,0){\circle{.2}}
	\put(2,0){\circle{.2}}
	\put(4,0){\circle{.2}}
\end{picture}
\caption{The cylinder, $\Z\mathop\square C_4$.}\label{fig:cyl}
\end{center}
\end{figure}

Given any o.n.b. $(\phi_n)$ for $\cA$ on the $N$-path, consider the bases
\[
w_1 = \frac{1}{2}(1,1,-1,-1)\,, \quad w_2 = \frac{1}{\sqrt{2}}(0,1,0,-1)\,, \quad w_3 = \frac{1}{\sqrt{2}}(1,0,-1,0)\,, \quad w_4=(1,1,1,1)
\]
and
\[
\kappa_j = \frac{1}{2}(1,\omega^j,\omega^{2j},\omega^{3j})
\]
for $\cA_{C_4}$, where $\omega=\ee^{\pi \ii/2}$ and $j=0,\dots,3$. By Proposition~\ref{prp:cart}, we know that the orthonormal eigenbases of $\Gamma$ approach some weighted averages.

If we choose the eigenbasis $\psi_{n,j} = \phi_n\otimes w_j$, then by \eqref{e:cartav},
\[
\langle \psi_{n,j},\opn(\overline{a})\psi_{n,j}\rangle =\begin{cases} \frac{1}{4}\sum_{q=1}^4 \langle a(\cdot+v_q)\rangle & \text{if } j=1,4\\ \frac{\langle a(\cdot+v_2)\rangle + \langle a(\cdot+v_4)\rangle}{2} & \text{if }j=2,\\ \frac{\langle a(\cdot +v_1) \rangle+\langle a(\cdot +v_3)\rangle}{2} & \text{if } j=3. \end{cases}
\]
On the other hand, if $\widetilde{\psi}_{n,j} = \phi_n\otimes \kappa_j$, then for $j=1,\dots,4$,
\[
\langle \widetilde{\psi}_{n,j}, \opn(\overline{a})\widetilde{\psi}_{n,j}\rangle = \frac{1}{4}\sum_{q=1}^4 \langle a(\cdot+v_q)\rangle \,.
\]

This example shows that $\langle \psi_u^{(N)},\opn(\overline{a})\psi_u^{(N)}\rangle$ in general depends on the choice of the basis, even for simple regular graphs, and it may or may not be the uniform average. In fact, this gives the uniform average for the basis $\widetilde{\psi}_{n,j}$, but not for $\psi_{n,j}$, take for example the observable
\[
a(i+v_1)=a(i+v_3)=-1 \,, \qquad a(i+v_2)=a(i+v_4)=1\,,
\]
where we parametrized the vertices of the cylinder $\Z\mathop\square C_4$ by $u = i+v_q$, where $i\in \Z$ is the layer's level and $v_q\in C_4= (v_1,v_2,v_3,v_4)$ are the vertices within it.

The problem with $\psi_{n,j}$ is that it is concentrated on half the cylinder for $j=2,3$, while $\widetilde{\psi}_{n,j}$ is spread on the whole. The semi-delocalization of $\psi_{n,j}$ is not detected by locally constant observables.

\subsection{Boxes again}\label{sec:boxes}
Back to Figure~\ref{fig:box}, let us show that the graph becomes quantum ergodic once we add a potential $(Q,-Q)$ which is copied across the layers, for any $Q>0$. 

In fact, in this case we get the Floquet eigenvalues
\[
E_\pm(\theta_\fb) = 2\cos 2\pi\theta \pm \sqrt{(1+2\cos 2\pi\theta)^2+Q^2}\,.
\]
We now use the idea in \S\ref{sec:1d}: if for some $s,w$ we have $E_s(\theta_\fb+\alpha_\fb)=E_w(\theta_\fb)$, then $\lambda=E_s(\theta_\fb+\alpha_\fb)$ solves both the characteristic polynomial of $H(\theta_\fb)$ and $H(\theta_\fb+\alpha_\fb)$. Denote $c_\theta:=2\cos 2\pi\theta$. It follows that for such $\lambda$,
\[
\lambda^2-2c_{\theta+\alpha}\lambda-(1+Q^2+2c_{\theta+\alpha}) = \lambda^2-2c_\theta\lambda-(1+Q^2+2c_\theta)
\]
In turn, this implies
\[
(c_{\theta+\alpha}-c_\theta)(\lambda+1)=0\,.
\]
So either $\lambda=-1$ or $c_{\theta+\alpha}-c_\theta=0$. The case $\lambda=-1$ never happens. In fact, if $\lambda = c+\sqrt{(1+c)^2+Q^2}$, then one can easily show that there is an $M$ such that $\lambda \ge -1+\frac{Q}{M}>-1$. Similarly, if $\lambda = c-\sqrt{(1+c)^2+Q^2}$, then we can find $M$ such that $\lambda\le -1-\frac{Q}{M}<-1$. 

Thus, the only way the Floquet assumption can be violated is if $c_{\theta+\alpha}-c_\theta=0$. Clearly, for a given nonzero $\alpha$, only $\theta=\frac{1-\alpha}{2}$ is possible. In particular, \eqref{e:eigenass} is satisfied.

\section{Complementary results}\label{sec:comp}

\subsection{QUE and eigenvector correlators}\label{sec:quecor}
\subsubsection{Quantum unique ergodicity}
We first investigate QUE for $\cA_\Z$ and $\cA_{\Z^2}$.

For $\Gamma =\Z$, taking $\LL_N$ with periodic conditions amounts to considering $N$-cycles. On $C_{4N}$, consider the observable $a_N = (1,0,1,0,\dots,1,0)$ and the eigenvector $v^{(N)}=\frac{1}{\sqrt{2N}}(0,1,0,-1,\dots,0,1,0,-1)$ with eigenvalue $0$, where the string $(0,1,0,-1)$ is repeated $N$ times. Then $\langle v^{(N)},a_N v^{(N)}\rangle=0$ while $\langle a_N\rangle=\frac{1}{2}$, so \eqref{e:que} is violated.

On $\Z^2$, the whole sequence may converge to a nonzero limit. If $e_\ell^{(N)}(k) = \frac{1}{N}\ee^{2\pi\ii k\cdot \ell/N}$, take $\phi_{(\ell_1,\ell_1)}= e_{(\ell_1,\ell_1)}$ and $\phi_{(\ell_1,\ell_2)} = \frac{1}{\sqrt{2}} e_{(\ell_1,\ell_2)}^{(N)} + \sgn(\ell_1-\ell_2)\frac{1}{\sqrt{2}} e_{(\ell_2,\ell_1)}^{(N)}$ if $\ell_1\neq \ell_2$. This gives an orthonormal eigenbasis with $|\phi_{(\ell_1,\ell_2)}^{(N)}(n)|^2 = \frac{1\pm\cos 2\pi[(\ell_1-\ell_2)(n_1-n_2)/N]}{N^2}$ if $\ell_1\neq \ell_2$. So $\langle \phi_{(\ell_1,\ell_2)},a_N\phi_{(\ell_1,\ell_2)}\rangle = \langle a_N\rangle \pm \frac{1}{N^2}\sum_n a_N(n)\cos 2\pi\frac{(\ell_1-\ell_2)(n_1-n_2)}{N}$. If $a_N(n) = f(n/N)$, we thus get
\[
\langle \phi_{(\ell_1,\ell_2)}^{(N)}, a_N \phi_{(\ell_1,\ell_2)}^{(N)}\rangle -\langle a_N\rangle \to \pm \int_{[0,1]^2} f(x,y)\cos 2\pi(\ell_1-\ell_2)(x-y)\,\dd x \dd y\,.
\]
This is nonzero for $f(x,y) = \cos 2\pi(\ell_1-\ell_2)(x-y)$. 

\subsubsection{No correlator universality}
We next consider the question of matrix observables. 

On $\Z^2$, consider standard basis $(e_\ell^{(N)})_{\ell\in \LL_N^2}$ and the basis $(\phi_\ell^{(N)})_{\ell\in\LL_N^2}$ defined in the previous paragraph. Consider
\[
K_N(n,m)=\begin{cases} 1& \text{if } n-m=(\pm 1,0),\\ 0 &\text{otherwise.} \end{cases}
\]

Then $\langle e_\ell^{(N)},K e_\ell^{(N)}\rangle = 2\cos (\frac{2\pi\ell_1}{N})$, so $\frac{1}{N^2}\sum\limits_{\ell\in\LL_N^2} |\langle e_\ell^{(N)},K e_\ell^{(N)}\rangle|^2 \to \int_0^1 4\cos^2(2\pi x)\,\dd x=2$. 

On the other hand, $\langle \phi_\ell^{(N)}, K\phi_\ell^{(N)}\rangle= \cos (\frac{2\pi \ell_1}{N}) + \cos(\frac{2\pi \ell_2}{N})$, so $\frac{1}{N^2}\sum\limits_{\ell\in\LL_N^2} |\langle \phi_\ell^{(N)},K\phi_\ell^{(N)}\rangle|^2 \to \int_{[0,1]^2} \cos^2(2\pi x_1)+\cos^2(2\pi x_2)+ 2\cos(2\pi x_1)\cos(2\pi x_2)\,\dd x = 1$.

This implies there can be no quantity $\langle K_N\rangle_{\lambda_j^{(N)}}$ independent of the basis such that $\frac{1}{N^2}\sum_j |\langle \psi_j^{(N)},K_N\psi_j^{(N)}\rangle - \langle K_N\rangle_{\lambda_j^{(N)}}|^2 \to 0$.

\subsubsection{Matrix generalization}
We finally sketch how to generalize quantum ergodicity to matrix observables $K$. For simplicity we only discuss the case $\nu=1$. We may assume $V_f = \{0\}$ up to translating coordinates. Here, $H(\theta_\fb)=E(\theta_\fb)$.

For Step 1, we note that
\[
(\ee^{\ii tH_N} K\ee^{-\ii tH_N}\psi)(k_\fa) = \sum_{r\in\LL_N^d} \ee^{\ii t E(\frac{r_\fb}{N})}(UK\ee^{-\ii tH_N}\psi)_re_r^{(N)}(k)\,.
\]
Here, $(UK\ee^{-\ii tH_N}\psi)_r = \frac{1}{N^{d/2}}\sum_{n\in\LL_N^d}\ee^{-\frac{\ii r_\fb}{N}\cdot n_\fa}(K\ee^{-\ii tH_N}\psi)(n_\fa)$. If $R$ is the band width, then $(K\ee^{-\ii tH_N}\psi)(n_\fa) = \sum_{|\tau|\le R} K(n_\fa,n_\fa+\tau_\fa)(\ee^{-\ii tH_N}\psi)(n_\fa+\tau_\fa)$. Denote $K^\tau(n_\fa):= K(n_\fa,n_\fa+\tau_\fa)$. Next, expand $K^\tau(n_\fa) = \frac{1}{N^{d/2}}\sum_{m\in \LL_N^d} K_m^\tau \ee^{\frac{\ii m_\fb\cdot n_\fa}{N}}$, where $K_m^\tau = \langle e_m^{(N)},K^\tau(\cdot_\fa)\rangle_{\ell^2(\LL_N^d)}$. Then we obtain
\begin{align*}
(UK\ee^{-\ii tH_N}\psi)_r &= \frac{1}{N^{d}}\sum_{n,m\in\LL_N^d}\sum_{|\tau|\le R}\ee^{-\frac{\ii(r_\fb-m_\fb)\cdot n_\fa}{N}} K_m^\tau(\ee^{\ii tH_N}\psi)(n_\fa+\tau_\fa)\\
&= \frac{1}{N^{d/2}} \sum_{m\in\LL_N^d} \sum_{|\tau|\le R} K_m^\tau \ee^{\frac{\ii (r_\fb-m_\fb)\cdot\tau_\fa}{N}}(U\ee^{-\ii tH_N}\psi)_{r-m} \,.
\end{align*}

From here, we proceed as before, replacing $a_m^{(N)}(v_q)$ by $\sum_{|\tau|\le R} K_m^\tau \ee^{\frac{\ii r_\fb \cdot\tau_\fa}{N}}$. There are of course many simplifications because $\nu=1$. In the end, $\overline{a}$ is replaced by
\[
\overline{K} = \sum_{|\tau|\le R} K_0^\tau \ee^{\frac{\ii r_\fb \cdot\tau_\fa}{N}} e_0^{(N)}(k) = \sum_{|\tau|\le R} \langle K^\tau\rangle \ee^{\frac{\ii r_\fb \cdot\tau_\fa}{N}} \,,
\]
where $\langle K^\tau\rangle = \frac{1}{N^d}\sum_{n\in\LL_N^d} K(n,n+\tau)$. Hence,
\begin{align*}
\langle \psi, \opn(\overline{K})\psi\rangle &= \sum_{k\in\LL_N^d}\overline{\psi(k_\fa)}\sum_{r\in\LL_N^d}(U\psi)_r\sum_{|\tau|\le R} \langle K^\tau\rangle \ee^{\frac{\ii r_\fb \cdot\tau_\fa}{N}} e_r^{(N)}(k)\\
&=\sum_{|\tau|\le R}\langle K^\tau\rangle\sum_{k\in\LL_N^d} \overline{\psi(k_\fa)}\psi(k_\fa+\tau_\fa) = \sum_{|\tau|\le R} \langle K^\tau\rangle \langle\psi,\psi(\cdot+\tau_\fa)\rangle\,.
\end{align*}
This is the same expression we stated in \S\ref{sec:eic}. Interestingly, by examining the proof, we see that $R$ can be taken to increase with $N$, like $R\lesssim N^\delta$ with $\delta<\frac{1}{2d}$.

\subsection{Bloch's theorem}\label{sec:bloch}
We prove here a version of the Bloch theorem for discrete periodic operators. This result is well-known in the continuum, but doesn't seem to have been explored in our setting. We also comment on the corresponding eigenfunction average.

\begin{thm}
Let $H$ be a periodic Schr\"odinger operator over the infinite periodic graph $\Gamma$, and suppose $\lambda\in \sigma(H)$. Then we may find $\Psi_{\lambda}$ on $\Gamma$ such that $H\Psi_{\lambda}=\lambda\Psi_{\lambda}$ and $\Psi_{\lambda}(k_{\mathfrak{a}}+v_n) = \ee^{\ii\theta_{\mathfrak{b}}\cdot k_{\mathfrak{a}}}f(v_n)$, for some $\theta\in [0,1)^d$ and $f$ on $V_f$.

Similarly, if $\lambda\in \sigma(H_{N})$, we may find $\Psi_{\lambda}$ on $\Gamma_N$ such that $H_N\Psi_\lambda= \lambda\Psi_\lambda$ and $\Psi_\lambda(k_\fa+v_n)=\ee^{\ii \frac{j_\fb\cdot k_\fa}{N}} f(v_n)$, for some $j\in \LL_N^d$ and $f$ on $V_f$.
\end{thm}
\begin{proof}
$H$ is unitarily equivalent to $\int_{\cC_{\mathfrak{b}}}H(\theta_{\mathfrak{b}})\,\dd\rho_\star$, so $\sigma(H) = \cup_{n=1}^\nu \sigma_n$, where $\sigma_n = E_n(\cC_{\mathfrak{b}})=[E_n^-,E_n^+] $, see \cite{BoSa,KorSa}. Hence, $\lambda\in \sigma(H)$ implies $\lambda = E_r(\theta_{\mathfrak{b}})$ for some $r$ and $\theta_{\mathfrak{b}}\in \cC_{\mathfrak{b}}$. Let $\psi_r^{\theta_{\mathfrak{b}}}$ be the corresponding eigenvector on $V_f$ and define $\Psi_{\lambda}(k_{\mathfrak{a}}+v_n):=\ee^{\ii\theta_{\mathfrak{b}}\cdot k_{\mathfrak{a}}}\psi_r^{\theta_{\mathfrak{b}}}(v_n)$. Then
\begin{align*}
H\Psi_{\lambda}(k_{\mathfrak{a}}+v_n) &= \sum_{u\sim k_{\mathfrak{a}}+v_n}\Psi_{\lambda}(u) + Q(v_n)\Psi_{\lambda}(k_{\mathfrak{a}}+v_n) = \sum_{w\sim v_n} \Psi_{\lambda}(w+k_{\mathfrak{a}})+Q(v_n)\Psi_{\lambda}(k_{\mathfrak{a}}+v_n)\\
&=\sum_{w\sim v_n} \Psi_{\lambda}(k_{\mathfrak{a}}+\lfloor w\rfloor_{\mathfrak{a}}+\{w\}_{\mathfrak{a}}) + Q(v_n)\Psi_{\lambda}(k_{\mathfrak{a}}+v_n) \\
&= \ee^{\ii\theta_{\mathfrak{b}}\cdot k_{\mathfrak{a}}}\Big(\sum_{w\sim v_n}\ee^{\ii\theta_{\mathfrak{b}}\cdot \lfloor w\rfloor_{\mathfrak{a}}}\psi_r^{\theta_{\mathfrak{b}}}(\{w\}_{\mathfrak{a}}) + Q(v_n) \psi_r^{\theta_{\mathfrak{b}}}(v_n)\Big) \\
&= \ee^{\ii\theta_{\mathfrak{b}}\cdot k_{\mathfrak{a}}} (H(\theta_{\mathfrak{b}})\psi_r^{\theta_{\mathfrak{b}}})(v_n) = \ee^{\ii\theta_{\mathfrak{b}}\cdot k_{\mathfrak{a}}} E_r(\theta_{\mathfrak{b}})\psi_r^{\theta_{\mathfrak{b}}}(v_n) = \lambda\Psi_{\lambda}(k_{\mathfrak{a}}+v_n)\,.
\end{align*}

The case of $\Gamma_N$ is the same since $H_N \equiv \mathop\oplus_{j\in\LL_N^d} H(\frac{j_\fb}{N})$.
\end{proof}

Note that on $\Gamma_N$, we have $\|\Psi_\lambda\|^2 = \sum_{k\in\LL_N^d}\sum_{n=1}^\nu |f(v_n)|^2 = N^d \|f\|_{\C^\nu}^2$. If $\widetilde{\Psi}_\lambda = \frac{1}{\|\Psi_\lambda\|}\Psi_\lambda$, then $\langle \widetilde{\Psi}_\lambda, a\widetilde{\Psi}_\lambda\rangle = \frac{1}{N^d\|f\|_{\C^\nu}^2}\sum_{k\in \LL_N^d}\sum_{n=1}^\nu | f(v_n)|^2 a(k_\fa+v_n) = \sum_{n=1}^\nu \langle a(\cdot +v_n)\rangle \frac{|f(v_n)|^2}{\|f\|_{\C^\nu}^2}$. This average is in general not uniform unless $a$ is locally constant. This is in accord with Theorem~\ref{thm:maingen}.

\begin{rem}
Note that these Bloch functions exist even in case of flat bands. For example, in Figure~\ref{fig:box}, instead of considering the localized functions $(\dots,0,0,\binom{1}{-1},0,0,\dots)$, one can consider the Bloch function $\ee^{2\pi\ii k\cdot \theta}\binom{1}{-1}$, where $k\in \Z$ is the position. We see that this delocalized function is also an eigenvector with the same eigenvalue $-1$. 

This shows the limitations of this theorem; while there always exist an eigenfunction with periodic modulus (hence well distributed over the crystal and delocalized), there can also exist a lot of localized eigenfunctions for the same energy, which is the phenomenon that quantum ergodicity investigates.
\end{rem}

\subsection{Bloch varieties and assumption \eqref{e:eigenass}}\label{sec:irred}
Let $p(\theta;\lambda)$ be the characteristic polynomial of $H(\theta_\fb)$. Let $z_j = \ee^{2\pi\ii \theta_j}$. By definition \eqref{e:htheta}, we see that $p(z;\lambda)$ is a Laurent polynomial in $z$ and polynomial in $\lambda$.

We say that $p$ is \emph{irreducible} if the only way to write it as a product of two Laurent polynomials $p(z;\lambda)=f(z;\lambda)g(z;\lambda)$ is to take $f$ or $g$ to be a Laurent monomial $Cz_1^{a_1}\cdots z_d^{a_d}$, for some $a_j\in \Z$. 

The important point in the previous definition is that the factors $f,g$ should be Laurent polynomials of $(z;\lambda)$. For example, as we saw in \eqref{e:charpol1d}, for Schr\"odinger operators with a periodic potential on $\Z$, we have $p(z;\lambda)=\Delta(\lambda)-z-z^{-1}$. In this case, studying irreducibility is equivalent to considering the polynomial
\begin{equation}\label{e:poly1d}
z^2-z\Delta(\lambda)+1\,.
\end{equation}
In principle one can always write this as a product $(z-g_1(\lambda))(z-g_2(\lambda))$. However, \eqref{e:poly1d} is actually regarded as irreducible here because $g_i(\lambda)$ are not polynomials of $\lambda$, cf. \cite[p.19]{GKT}.

If a flat band $E_r(\theta_\fb)\equiv c$ exists, then the characteristic polynomial is reducible, since we then have $p(z;\lambda) = (\lambda-c)g(z;\lambda)$ for some Laurent polynomial $g(z;\lambda)$. Thus, irreducibility implies pure ac spectrum.

Irreducibility entails that the Bloch variety of $H$,
\[
B_H = \{(\theta,\lambda)\in \C^{d+1}: p(z;\lambda)=0\}
\]
cannot be written as the union of two proper analytic subsets, except for periodicity. That is, if $\Omega_1$ and $\Omega_2$ are two components of $B_H$, then $\Omega_2=\Omega_1+(k,0)$ for some $k\in \Z^d$.

\smallskip
	
Now, let us write
\[
p(z;\lambda) = (-1)^\nu\prod_{m=1}^K p_m(z;\lambda)
\]
for some irreducible Laurent polynomials $p_m(z;\lambda)$. It is proved in \cite{Wen2} that if for all nonzero $\alpha\in[0,1)^d$ and all $m_1,m_2$,
\begin{equation}\label{e:wencai}
p_{m_1}(z;\lambda) \not\equiv p_{m_2}(\zeta z;\lambda)
\end{equation}
as Laurent polynomials, where $\zeta=(\ee^{2\pi\ii\alpha_1},\dots,\ee^{2\pi\ii\alpha_d})$ and $\zeta z:=(\zeta_1 z_1,\dots,\zeta_d z_d)$, then \eqref{e:eigenass} is satisfied. In particular, if $p(z;\lambda)$ is irreducible and for any $\zeta\neq 1$ with $|\zeta|=1$, we have $p(z;\lambda)\not\equiv p(z\zeta;\lambda)$ as polynomials, then \eqref{e:eigenass} is satisfied. This is a remarkable simplification as we now only need to study the condition for the characteristic polynomial, instead of the eigenvalues which may be difficult to compute or may have complicated root expressions. This is in fact how \eqref{e:eigenass} is established in \cite{Wen2}, using \cite{Wencai}.

For comparison, to establish the criterion in general, we can always argue as in \S\ref{sec:1d}, namely try to show that for \emph{fixed} $\lambda$, there are not too many $z$ such that $p(z,\lambda)=p(z\zeta,\lambda)$. In case of irreducibility however, we just need to show that $p(z,\lambda)\not \equiv p(z\zeta,\lambda)$ as polynomials. This can be done for example by comparing the coefficients of $\lambda^k$ for some $k$ and showing they can only be equal on a set of zero measure.

In particular, the Bloch variety for periodic Schr\"odinger operators on the triangular lattice and the EHM lattice is also irreducible \cite{FLM}, so one only needs to verify $p(z;\lambda)\not\equiv p(z\zeta;\lambda)$. The argument used in \cite{Wen2} applies to Schr\"odinger operators with a periodic potential on the triangular lattice, so they are quantum ergodic as well.

It should be noted that irreducibility is not necessary for \eqref{e:eigenass} to hold. For example, the infinite ladder \S\ref{sec:lad} has characteristic polynomial $(z+z^{-1}-\lambda)^2-1 = (z+z^{-1}+1-\lambda)(z+z^{-1}-1-\lambda)$, hence reducible. Still, \eqref{e:eigenass} is satisfied.

Even when the characteristic polynomial is reducible, criterion \eqref{e:wencai} applies, and it can be much simpler to check than \eqref{e:eigenass} directly.\footnote{Note that we only discussed the (ir)reducibility of the \emph{Bloch} variety here. The irreducibility of the \emph{Fermi} variety, where $\lambda$ is fixed, is significantly harder to prove \cite{Wencai}, but we do not need it.}

\subsection*{Acknowledgments}
M.S. is very thankful to Wencai Liu for discussions concerning his results \cite{FLM,Wencai,Wen2} on irreducibility of Bloch varieties. 

\providecommand{\bysame}{\leavevmode\hbox to3em{\hrulefill}\thinspace}
\providecommand{\MR}{\relax\ifhmode\unskip\space\fi MR }
\providecommand{\MRhref}[2]{%
  \href{http://www.ams.org/mathscinet-getitem?mr=#1}{#2}
}
\providecommand{\href}[2]{#2}

\end{document}